\documentclass[10pt,twocolumn,letterpaper]{article}

\usepackage[textwidth=6.75in,textheight=9.25in,centering,columnsep=0.25in]{geometry}
\usepackage[T1]{fontenc}
\usepackage{lmodern}
\makeatletter
\renewcommand\section{\@startsection{section}{1}{\z@}%
  {-3.5ex \@plus -1ex \@minus -.2ex}%
  {2.3ex \@plus .2ex}{\normalfont\large\bfseries\raggedright}}
\renewcommand\subsection{\@startsection{subsection}{2}{\z@}%
  {-3.25ex \@plus -1ex \@minus -.2ex}%
  {1.5ex \@plus .2ex}{\normalfont\normalsize\bfseries\raggedright}}
\renewcommand\subsubsection{\@startsection{subsubsection}{3}{\z@}%
  {-3.25ex \@plus -1ex \@minus -.2ex}%
  {1.5ex \@plus .2ex}{\normalfont\normalsize\bfseries\raggedright}}
\makeatother

\usepackage[round]{natbib}

\usepackage{amsmath,amssymb,amsthm,array,booktabs,float,graphicx,longtable,mathtools,placeins,dblfloatfix}
\usepackage{tikz}
\usetikzlibrary{arrows.meta,positioning,calc}
\usepackage[hidelinks]{hyperref}

\newtheorem{proposition}{Proposition}
\newcommand{\KL}{\operatorname{KL}}
\newcommand{\TV}{\operatorname{TV}}
\newcommand{\E}{\mathbb{E}}
\newcommand{\R}{\mathbb{R}}
\newcommand{\Mcal}{\mathcal{M}}
\newcommand{\method}{SM-VTI}
\newcommand{\vtibaseline}{VTI}

\newcommand{\best}[1]{\ensuremath{\mathbf{#1}}}

\newtoks\supplementarymaterial
\long\def\MoveToAppendix#1{%
  \global\supplementarymaterial\expandafter{\the\supplementarymaterial #1}%
}

\title{Structured Dimension-Matching\\
Variational Transdimensional Inference}
\author{Pingping Yin \qquad Xiyun Jiao\\[0.4em]
\small Department of Statistics and Data Science\\
\small Southern University of Science and Technology, Shenzhen, China\\[0.3em]
\small Correspondence: \texttt{jiaoxy@sustech.edu.cn}}
\date{}

\hypersetup{
  pdftitle={Structured Dimension-Matching Variational Transdimensional Inference},
  pdfauthor={Pingping Yin and Xiyun Jiao},
  pdfsubject={Research preprint},
  pdfkeywords={variational inference, transdimensional inference, normalizing flows, Bayesian model selection}
}

\begin{document}
\maketitle
\raggedbottom

\begin{abstract}
Variational transdimensional inference (VTI) has extended flow-based variational
inference to transdimensional Bayesian inference via CoSMIC normalizing flows.
CoSMIC enables a single transdimensional flow by saturating all models to the
maximum dimension. We introduce structured
dimension-matching variational transdimensional
inference (\method{}), which defines a joint approximation directly in
the model-specific spaces. A construction tree gives each model a unique
path. Each edge first copies designated shared coordinates and constructs new
coordinates through an invertible map, then applies a conditional flow.
Locally normalized stopping and
child-selection probabilities define model masses, while a shared conditional
flow learns the remaining geometry. \method{} supports reverse-KL training by
exact enumeration or stochastic path sampling. Across robust variable
selection, phylogenetic model selection, and nonlinear directed acyclic graph
(DAG) structure learning, the comparison with VTI depends on the approximation
metric and training configuration. Learned local maps improve approximation in
our variable-selection ablation. In five paired Primates starts with shared
training, \method{} has lower conditional and joint negative log density
in every pair. VTI has lower mean model-mass error. Both approaches support
nonenumerative training over $2^{24}$ subset models.
\end{abstract}

\section{Introduction}
Transdimensional inference~\citep{green1995,green2003transdimensional,hastiegreen2012,sisson2005transdimensional}
commonly arises in applications, such as variable
selection~\citep{george1993}, phylogenetics~\citep{huelsenbeck2004}, and geoscientific
inversion~\citep{J_PTRSA_sambridge2013transdimensional}.
Given observed data $\mathbf y$, denote the model index by $m\in \mathcal M$ and the parameters within each model by $\theta_m\in\Theta_m\subset\mathbb R^{d_m}$. Under the Bayesian framework, the target distribution is the posterior on the joint space $\Omega=\bigcup_{m\in\mathcal M}(\{m\}\times\Theta_m)$:
\begin{equation}
p(m,\theta_m\mid\mathbf y)\propto
p(\mathbf y\mid\theta_m,m)p(\theta_m\mid m)p(m),
\label{eq:trans_model_posterior}
\end{equation}
where $p(\mathbf y\mid\theta_m,m)$ is the model-specific likelihood, $p(\theta_m\mid m)$ and $p(m)$ are the priors of the within-model parameters and model index respectively. 

Reversible jump Markov chain Monte Carlo~\citep[RJMCMC,][]{green1995} is one of the classical methods for
sampling from the posterior in \eqref{eq:trans_model_posterior}. It applies the \emph{dimension-matching} strategy to move between models with different dimensions. Specifically, to generate a proposal from model $m$ to $m'$, RJMCMC augments both $\theta_m$ and $\theta_{m'}$ to the same dimension using auxiliary variables and then applies
bijective maps to transform between the augmented parameter spaces.
The efficiency of RJMCMC depends on the choice of auxiliary variables, bijective maps, and model-index proposals. Many between-model proposals have been introduced~\citep{green1995,richardson1997,brooks2003,davies2023}. Transport reversible jump (TRJ), developed by \citet{davies2023}, provides a general strategy for constructing between-model proposals. TRJ learns a transport between the conditional posterior $p(\theta_m\mid\mathbf y, m)$
and a reference distribution, and if exact transports are used, the acceptance probability of the resulting RJMCMC depends only on
the model probabilities. \citet{yin2025vitrj} uses the reverse Kullback-Leibler (KL) divergence to train the transports so that pilot posterior samples are not required. 

We use variational inference with evaluable flow densities. More generally,
variational inference also admits implicit families~\citep{titsias2019implicit}.
\citet{davies2025} embed each $\Theta_m$ in dimension
$d_\text{max}=\max_m d_m$ and use one CoSMIC flow with a model-dependent
activity mask. The active and auxiliary densities then factorize.
An encoded construction path can be supplied to its conditioner.
Structured Dimension-Matching Variational Transdimensional Inference
(\method{}) uses the construction relation to specify the order and
coordinate structure of the transport itself. We define
$
q_{\psi,\phi}(m,\theta_m)=q_\psi(m)q_\phi(\theta_m\mid m)
$
directly on each model-specific space $\{m\}\times\Theta_m$. The tree gives each
model a unique root-to-node path. Along a path, an invertible
dimension-matching map at each edge transforms the original parameters together
with auxiliary variables into parameters of the child model. Finally, a shared
coupling flow learns the conditional geometry in
$q_\phi(\theta_m\mid m)$. At each tree node, whether the path stops (at a model node) or continues to a child is decided by locally normalized probabilities, and the product of the selected probabilities along the path gives $q_\psi(m)$. We train the variational density $q_{\psi,\phi}(m,\theta_m)$ using a reverse-KL objective. 

The contributions of SM-VTI are threefold:
\begin{enumerate}
\setlength{\itemsep}{0pt}
\item We define a native-dimensional variational family by composing typed,
invertible maps along a construction tree. Each dimension-matching map copies
designated coordinates at the input of the child flow and constructs new
coordinates from auxiliary variables.
\item In SM-VTI, a shared
flow learns conditional geometry and locally normalized tree probabilities define model masses. Path-specific variational densities enable reverse-KL
training by enumeration or stochastic path sampling, with an unbiased
leave-one-out policy gradient.
\item We evaluate posterior approximation and graph recovery across several
applications. Ablations examine the learned dimension-matching maps,
construction order, model-policy parameterization, and optimization
configuration. The $2^{24}$-model experiment evaluates quality and cost when
both methods use nonenumerative training.
\end{enumerate}

\section{Background}
\label{sec:background}

\subsection{Normalizing flows}
A normalizing flow is a generative model that learns an invertible mapping between a simple probability distribution and a complex data distribution. 
Both autoregressive and coupling-based approaches are commonly used to construct normalizing flows. An autoregressive flow with a MADE conditioner~\citep{germain2015made} is defined element-wise as
\begin{equation}
T(\boldsymbol{z})_i = \tau(z_i;\boldsymbol{\zeta}_i(\boldsymbol{z}_{<i};\omega_{\mathrm{AF}})),\qquad i=1,\ldots,d,
\label{eq:AF-transform}
\end{equation}
where $\tau(\cdot;\boldsymbol{\xi})$ is a univariate diffeomorphism with
admissible parameters $\boldsymbol{\xi}\in\boldsymbol{\Xi}$.
$\boldsymbol{\zeta}_i$ is the $i$th parameter block of the masked
conditioner with weights $\omega_{\mathrm{AF}}$. Its input dependence is
restricted to $\boldsymbol{z}_{<i}$.

The coupling-based flow model, such as RealNVP~\citep{dinh2017realnvp}, is constructed by stacking a sequence of \emph{affine coupling layers} (ACL). Given an input vector $\boldsymbol{z}$ of dimension $d$ and $d_0 < d$, the output $\boldsymbol{x}$ of an ACL has the form of
\begin{equation}
\begin{array}{l}
\boldsymbol{x}_{1:d_0} = \boldsymbol{z}_{1:d_0}, \\
\boldsymbol{x}_{d_0+1:d} = \boldsymbol{z}_{d_0+1:d} \odot \exp\bigl(s(\boldsymbol{z}_{1:d_0})\bigr) + t(\boldsymbol{z}_{1:d_0}),
\end{array}
\label{eq:realnvp-acl}
\end{equation}
where $s$ and $t$ are scale and translation functions from $\mathbb{R}^{d_0} \to \mathbb{R}^{d-d_0}$, and $\odot$ is the element-wise product. The Jacobian matrix $\frac{\partial \boldsymbol{x}}{\partial \boldsymbol{z}}$ is lower-triangular and the determinant is $\exp\left[\sum_j s(\boldsymbol{z}_{1:d_0})_j\right]$, which does not involve the Jacobian of either $s$ or $t$. Therefore $s$ and $t$ can be arbitrarily complex.

Conditional normalizing flows make the invertible mapping and the base distribution depend on another conditioning variable $c$ \citep{winkler2019}. The CoSMIC normalizing flow of~\citet{davies2025} is a context-masking extension of
conditional autoregressive flows to fit the transdimensional settings.

\subsection{Variational transdimensional inference}
\label{sec:background-avti}

Variational transdimensional inference~\citep[VTI,][]{davies2025} uses a saturated-space representation where all the
model-conditional parameter spaces are augmented to the maximum dimension
$d_\text{max} = \max_{m\in\Mcal} \{d_m\}$. Let $\eta_m(\theta_m)=p(\mathbf y\mid\theta_m,m)p(\theta_m\mid m)$ be the
unnormalized conditional posterior of model $m$. Given a reference density $\nu$, VTI considers the following saturated target: 
\begin{equation}
\widetilde\eta^{\mathrm{VTI}}(m,\theta_m,u_{\setminus m})
=p(m)\eta_m(\theta_m)\,
(\nu^{\otimes(d_\text{max}-d_m)})(u_{\setminus m}),
\label{eq:saturated-target}
\end{equation} 
where $u_{\setminus m}\in \mathbb R^{d_\text{max}-d_m}$ are the auxiliary variables in model $m$. 

VTI approximates the target in \eqref{eq:saturated-target} by the product of a model
distribution $q_\psi(m)$ and a density $\widetilde q_\phi(\theta_m,u_{\setminus m}\mid m)$ specified by a model-conditional flow $T_\phi(\cdot\mid m)$ in the saturated space, that is, $(\theta_m,u_{\setminus m})=T_\phi(z\mid m)$ for $z\sim\nu^{\otimes d_\text{max}}$. The joint
variational density is then
\begin{equation}
\widetilde q^{\mathrm{VTI}}_{\psi,\phi}
(m,\theta_m,u_{\setminus m})
=q_\psi(m)\widetilde q_\phi(\theta_m,u_{\setminus m}\mid m).
\label{eq:avti-joint}
\end{equation}
\citet{davies2025} propose the CoSMIC normalizing flows to construct $T_\phi(\cdot\mid m)$.


\paragraph{CoSMIC normalizing flow.}
CoSMIC (\emph{COntextually}-\emph{Specified Masking for Identity-mapped Components}) enables the construction of a single flow on the transdimensional support $\Omega=\bigcup_{m\in\mathcal M}(\{m\}\times\Theta_m)$. CoSMIC achieves the exact factorization of the original parameters and auxiliary variables as in \eqref{eq:saturated-target}, and uses context masks to apply
learned maps on the active coordinates and identity transformations
on the inactive coordinates.  The details of the coordinate-wise mask construction are given in Appendix~\ref{app:avti-details}. Finally, the variational density becomes
\begin{equation}
\widetilde q_\phi(\theta_m,u_{\setminus m}\mid m)
=q_\phi(\theta_m\mid m)\,
(\nu^{\otimes(d_\text{max}-d_m)})(u_{\setminus m}).
\label{eq:cosmic-factorization}
\end{equation}
The auxiliary reference factors in
\eqref{eq:saturated-target} and \eqref{eq:cosmic-factorization} cancel out from the
joint reverse-KL objective. VTI composes multiple CoSMIC layers for constructing $T_\phi(\cdot\mid m)$ and uses the masked affine-autoregressive or rational-quadratic-spline-autoregressive transform
with residual MADE blocks for each layer. Different active parameters occupy different positions in successive autoregressive transforms. The invertible map for the active coordinates is obtained by reversing the corresponding block between successive transforms.  

The induced variational density $q_\psi(m)q_\phi(\theta_m\mid m)$ is optimized with the joint reverse-KL as in \eqref{eq:joint-loss}. The choices for the model distribution and its score-function estimator given by \citet{davies2025} are shown in
Appendix~\ref{app:avti-details}.

\paragraph{Sequential construction for graph posteriors.}
\citet{deleu2023} use generative flow networks to grow a DAG and sample its
parameters conditionally. Our construction assigns a unique path to a
model state and composes typed invertible maps along it, yielding a joint
variational density for reverse-KL optimization.

\MoveToAppendix{
\subsection{VTI construction used for comparison}
\label{app:avti-details}

\paragraph{Joint objective and discrete model distribution.}
For the augmented target above, define the conditional contribution to the
unnormalized joint reverse-KL objective
\begin{equation}
 \begin{aligned}
 \ell_\phi(m)&=\E_{z\sim\nu^{\otimes d_\text{max}}}
 \bigl[\log q_\phi\{\theta_m(z)\mid m\}\\
 &\qquad-\log\eta_m\{\theta_m(z)\}\bigr],
 \end{aligned}
 \label{eq:avti-conditional-loss}
\end{equation}
where $\theta_m(z)$ is the active part of $T_\phi(z\mid m)$.  VTI's joint
reverse-KL, up to the marginal evidence constant, is
\begin{equation}
\mathcal L_{\mathrm{VTI}}(\psi,\phi)
=\E_{m\sim q_\psi}
\left[
 \ell_\phi(m)-\log p(m)+\log q_\psi(m)
\right].
\label{eq:avti-objective}
\end{equation}

The continuous parameters $\phi$ use standard reparameterization
gradients~\citep{kingma2014}.  VTI parameterizes $q_\psi$ by categorical
logits for a finite collection of models or, for combinatorial subset models,
by a MADE autoregressive distribution over a binary inclusion vector
$s(m)=(s_1,\ldots,s_{d_s})$,
\begin{equation}
 q_\psi(m)=\prod_{j=1}^{d_s}
 q_\psi\{s_j(m)\mid s_1(m),\ldots,s_{j-1}(m)\}.
\label{eq:avti-made-model-distribution}
\end{equation}
We use categorical logits in the fixed 128-model and PHYLO comparisons and
the MADE distribution in the cardinality study.  The Monte Carlo discrete
optimization uses score-function gradients with a control variate.  The
corresponding categorical estimator is given below; its control variate has
zero score-weighted expectation.

\paragraph{CoSMIC masking.}
CoSMIC first applies a model-dependent strict left-alignment permutation
$P_m$, which moves the active coordinates to the first $d_m$ positions while
preserving their relative order.  Let $A_i(m)\in\{0,1\}$ denote whether the
resulting $i$th coordinate is active.  An autoregressive network produces the
parameters $\rho_i$ of the coordinate-wise bijection $\tau_{\rho_i}$, and a
context-to-mask map $C_i(m)$ transfers the activity indicator to this
parameter block.  CoSMIC then combines the learned parameter with the fixed
identity parameter $\rho_{\mathrm{Id},i}$ according to
\begin{equation}
 \begin{aligned}
 \rho_i^{C_i(m)}&=\{1-C_i(m)\}\odot\rho_{\mathrm{Id},i}
   +C_i(m)\odot\rho_i ,\\
 \tau_{\rho_{\mathrm{Id},i}}(z^{(i)})&=z^{(i)}.
 \end{aligned}
 \label{eq:cosmic-mask}
\end{equation}
For an inactive coordinate, $C_i(m)=0$, so the layer is exactly the identity.
For example, an affine autoregressive coordinate has
$\tau_{\rho_i}(z^{(i)})=a_i+b_i z^{(i)}$ and identity parameter
$(a_i,b_i)=(0,1)$; the same construction applies to the parameters of a
rational-quadratic spline~\citep{durkan2019}.  Writing
$T_\phi^{(\ell)}(\cdot\mid m)=(\tau_{\rho^{C(m)}_1},\ldots,
\tau_{\rho^{C(m)}_{d_\text{max}}})(\cdot)$ for a masked autoregressive step, the full
left-aligned transport composes these steps with model-specific partial
reversals of the active coordinates, preceded by $P_m$ and followed by
$P_m^{-1}$.  In VTI, the transform
parameters are produced by residual MADE blocks: their masks enforce
autoregressive dependence, while the left alignment and partial reversals set
the coordinate order~\citep{davies2025,germain2015made}.

\paragraph{VTI score-function gradient.}
With $\widehat\ell_\phi$ the one-sample version of
\eqref{eq:avti-conditional-loss}, the discrete score-function gradient with
control variate $\varsigma$ is
\begin{equation}
 \nabla_\psi\mathcal L_{\mathrm{VTI}}
 =\E_{m,z}\!\Bigl[\big\{\widehat\ell_\phi(m,z)+\log q_\psi(m)-\log p(m)-\varsigma\big\}\times
 \nabla_\psi\log q_\psi(m)\Bigr].
\label{eq:avti-score-gradient}
\end{equation}
Here the unit score term from differentiating $q_\psi(m)\log q_\psi(m)$ is
omitted because its expectation is zero.
}

\section{Structured Dimension-Matching Variational Inference}
\label{sec:method}

\method{} defines $q_{\psi,\phi}(m,\theta_m)=q_\psi(m)q_\phi(\theta_m\mid m)$
on the model-specific supports. We minimize
$\KL\{q_{\psi,\phi}\,\|\,p(m,\theta_m\mid\mathbf y)\}$ over $(\psi,\phi)$.
After dropping the additive log evidence, the objective is
\begin{equation}
\begin{aligned}
\mathcal L_{\mathrm{joint}}(\psi,\phi)
&=\sum_{m\in\Mcal}q_\psi(m)\left[
\log\frac{q_\psi(m)}{p(m)}\right.\\
&\left.\qquad+\E_{q_\phi(\cdot\mid m)}
\log\frac{q_\phi(\theta_m\mid m)}{\eta_m(\theta_m)}\right].
\end{aligned}
\label{eq:joint-loss}
\end{equation}

\subsection{Tree-structured model spaces and dimension matching}

Let $\mathcal T=(\mathcal V,\mathcal E,r)$ be a finite tree with a designated
root $r$ and edges directed from parent to child. The experiment specifies
the parent rule, the state coordinates and the edge types before fitting.
A state is either a candidate model or an intermediate construction state.
Each $m\in\Mcal$ has exactly one model state $v_m$ in
$\mathcal V_{\mathrm M}\subseteq\mathcal V$. Every state reaches a model
state, including itself when it is a model state. A path can stop at a
model state and must continue at an intermediate state.
For example, if covariates are added in increasing index order, the parent
of a nonempty subset is obtained by deleting its largest index. This
canonical rule defines the tree without listing its vertices.
Figure~\ref{fig:construction-tree} shows the two-covariate case.
For each $m$, let $G_m=(r,e_1,\ldots,e_{L_m},v_m)$ be its unique path.
An edge $e=(v,c)$ has $d_c\geq d_v$ and auxiliary dimension $r_e=d_c-d_v$.
The edges are steps of the generative construction used to sample $q$.

\begin{figure}[t]
\centering
\begin{tikzpicture}[
  x=1cm,y=1cm,>=Latex,
  every node/.style={font=\fontsize{8.5}{10}\selectfont},
  state/.style={draw=black!60,rounded corners=2pt,fill=white,
    minimum height=8mm,minimum width=20mm,align=center},
  onpath/.style={state,draw=blue!60!black,fill=blue!5},
  construction/.style={-Latex,line width=.6pt,draw=black!60},
  selected/.style={construction,draw=blue!60!black,line width=.95pt}
]
\node[onpath] (empty) at (2.8,0) {$\varnothing$ (root)\\$d=0$};
\node[onpath] (first) at (1.25,-1.55) {$\{x_1\}$\\$d=1$};
\node[state] (second) at (4.35,-1.55) {$\{x_2\}$\\$d=1$};
\node[onpath] (both) at (1.25,-3.05) {$\{x_1,x_2\}$\\$d=2$};
\draw[selected] (empty.south west) -- node[above left,inner sep=1.8pt] {add $x_1$} (first.north east);
\draw[construction] (empty.south east) -- node[above right,inner sep=1.8pt] {add $x_2$} (second.north west);
\draw[selected] (first) -- node[right,inner sep=2.5pt] {add $x_2$} (both);
\node[align=left,inner sep=0pt] at (4.35,-2.75) {All nodes\\allow stopping.};
\end{tikzpicture}
\caption{\textbf{A branching construction tree.} Each nonempty subset has
one parent. The highlighted path generates $\{x_1,x_2\}$ by first adding
$x_1$. In this schematic the dimension is the number of optional
coefficients.}
\label{fig:construction-tree}
\end{figure}

For an edge $e=(v,c)$, let $z_v\in\R^{d_v}$ be the unconstrained parent
vector. The edge type $\tau_e$ is a fixed label for the prescribed edit
family, such as a birth, split, or identity. The corresponding
map can have learned parameters. Partition the parent coordinates into
retained $R_e$ and edited $S_e$. With $u_e\sim\rho_e$ on $\R^{r_e}$,
$T_{\tau_e}(\cdot;z_{v,R_e})$ bijectively maps the $|S_e|+r_e$ coordinates
of $(z_{v,S_e},u_e)$ to the corresponding child coordinates $D_e$.
The inverse permutation $P_e^{-1}$ restores the child's coordinate order:
\begin{align}
 \widetilde z_e=\mathcal B_e(z_v,u_e)
 =P_e^{-1}\begin{bmatrix}z_{v,R_e}\\T_{\tau_e}(z_{v,S_e},u_e;z_{v,R_e})\end{bmatrix}.
\label{eq:dimension-map}
\end{align}
The block triangular Jacobian gives
$J_e^{\mathcal B}=\log|\det \partial_{(z_{v,S_e},u_e)}T_{\tau_e}|$.
The inverse $\mathcal B_e^{-1}$ recovers both $z_v$ and $u_e$.
For a dimension-preserving identity edge, $S_e$ and $u_e$ are empty,
$\rho_e(u_e)=1$, and $\mathcal B_e$ and the edge flow are identities
with zero log determinants. The same zero-dimensional density convention
applies at a zero-dimensional root.

\subsection{Variational density induced by the path transport}

The map $\mathcal B_e$ realizes the prescribed coordinate edit. An invertible
conditional flow $F_{\omega,e}:\R^{d_c}\to\R^{d_c}$ is then used to
model the remaining geometry in the parameter space of the child model. For the edge $e=(v,c)$, 
\begin{equation}
 z_c=F_{\omega,e}(\widetilde z_e;h_e),
 \label{eq:edge-flow}
\end{equation}
where $h_e=H_\lambda(v,c,\tau_e,k_e)$ is an edge-specific context formed by the
parent and child states $(v,c)$, map type $\tau_e$, and edge depth $k_e$. The context map $H_\lambda$ may be fixed or learned; $\lambda$ denotes its parameters when it is learned. The maps $F_{\omega,e}$
share the same weights $\omega$ across edges. We construct $F_{\omega,e}(\cdot;h_e)$ using context-conditioned RealNVP. In each layer, the input is
partitioned as $z=(z_A,z_B)$ and mapped to
$z'=(z'_A,z'_B)$ via $z'_A=z_A$ and $z'_B=z_B\odot\exp\{s_\omega(z_A,h_e)\}+t_\omega(z_A,h_e)$. The inverse of the map is then $z_B=\{z'_B-t_\omega(z'_A,h_e)\}
\odot\exp\{-s_\omega(z'_A,h_e)\}$. Alternating masks can transform copied
coordinates as well. Thus the equality of retained coordinates applies
at $\widetilde z_e$, before this flow. The complete map retains the
reversible relationship $(z_v,u_e)=\mathcal B_e^{-1}\{F_{\omega,e}^{-1}(z_c;h_e)\}$.

\paragraph{Shared transport and structural inductive bias.}
Models with a common construction prefix reuse the same lifts and edge
flows along it, giving them a shared intermediate distribution.
Full-path CoSMIC passes the full path to a shared conditioner.
The finite architectures therefore share transformations in different ways.
With a common reference and unrestricted active and terminal maps,
both representations yield the same family of conditional densities
(Appendix~\ref{app:saturated-equivalence}).

\begin{figure}[t]
\centering
\begin{tikzpicture}[
  x=1cm,y=1cm,>=Latex,
  every node/.style={font=\fontsize{8.5}{10}\selectfont},
  block/.style={draw=black!55,rounded corners=2pt,
    fill=white,align=center,minimum height=8mm,minimum width=22mm},
  arrow/.style={-Latex,line width=.65pt,draw=black!75}
]
\node[block] (parent) at (0,0) {Parent $z_v\in\mathbb R^1$};
\node[block] (aux) at (3.6,0) {Auxiliary $u_e\in\mathbb R^1$};
\node[block,draw=blue!60!black,fill=blue!5,minimum width=43mm] (lift) at (1.8,-1.43)
  {$\widetilde z_e=\big(z_v,\,\mu(z_v)+s(z_v)u_e\big)\in\mathbb R^2$\\$s(z_v)>0$};
\node[block] (flow) at (1.8,-2.72) {Child $z_c\in\mathbb R^2$};
\node[inner sep=0pt] (more) at (1.8,-3.38) {$\vdots$};
\node[block] (output) at (1.8,-4.35) {Output $\theta_m\in\Theta_m$};
\draw[arrow] (parent.south) -- ([xshift=-13mm]lift.north);
\draw[arrow] (aux.south) -- ([xshift=13mm]lift.north);
\node[fill=white,inner sep=1pt] at (1.8,-.68) {dimension matching $\mathcal B_e$};
\draw[arrow] (lift) -- node[right,inner sep=4pt] {conditional flow $F_{\omega,e}$} (flow);
\draw[arrow] (more.south) -- node[right,inner sep=4pt] {terminal map to support} (output.north);
\end{tikzpicture}
\caption{\textbf{A one-coordinate birth.} The dimension-matching map sends
one parent coordinate and one auxiliary coordinate to two child coordinates.
The conditional flow acts on the resulting two-dimensional vector.
The dots indicate repeated edge steps before the terminal transformation
$\chi_m\circ F_m^{\mathrm{term}}$.}
\label{fig:dimension-matching}
\end{figure}

Recall that $G_m=(r,e_1,\ldots,e_{L_m},v_m)$ is the unique path from $r$ to model $m$. From the base distribution $b_0$ for the root, draw $\epsilon_0\sim b_0$ with support $\R^{d_r}$, and let
$F_r:\R^{d_r}\rightarrow\R^{d_r}$ be the invertible root flow.  Then the
path-wise transformation is
\begin{align}
 z_0&=F_r(\epsilon_0), &
 \widetilde z_{e_i}&=\mathcal B_{e_i}(z_{i-1},u_{e_i}),\nonumber\\
 z_i&=F_{\omega,e_i}(\widetilde z_{e_i};h_{e_i}), &
 \theta_m&=\chi_m\{F_m^{\mathrm{term}}(z_{L_m})\},
\label{eq:path-map}
\end{align}
for $i=1,\ldots,L_m$.  Here $z_{i-1}$ is the parameter vector for the parent of $e_i$,
and $z_{L_m}\in\R^{d_m}$ is the unconstrained parameter at the
designated state $v_m$. The terminal map
$F_m^{\mathrm{term}}:\R^{d_m}\to\R^{d_m}$ is a bijection that may share
weights across models and condition on terminal model metadata. Each edge
flow and the terminal flow may be set to the identity. The map
$\chi_m:\R^{d_m}\to\Theta_m$ maps the output of $F_m^{\mathrm{term}}$
to the constrained parameter $\theta_m$. It can be the identity, exponential, inverse-logit, or inverse additive-log-ratio (ALR) transformation, depending on the support of $\theta_m$~\citep{papamakarios2021,aitchison1986}.
Write the assembled map as $\theta_m=\Phi_{\phi,m}(\xi_m)$ with
$\xi_m=(\epsilon_0,u_{e_1},\ldots,u_{e_{L_m}})$.
The variable-selection cardinality sweep uses the terminal-flow special
case, with $F_{\omega,e_i}=\mathrm{Id}$ and $J_{e_i}^{F}=0$ before the
final coupling stack $F_m^{\mathrm{term}}$ in the completed chart.
The dimension-matching maps still execute along the path. The finite
128-model variable-selection and archived phylogenetic implementations
apply the shared flow after each edit and set
$F_m^{\mathrm{term}}=\mathrm{Id}$.

Since the complete path map is invertible, 
each $\theta_m\in\Theta_m$ uniquely determines $\epsilon_0$ and
$(u_{e_i})_{i=1}^{L_m}$. The resulting conditional density is
\begin{equation}
q_\phi(\theta_m\mid m)=\frac{b_0(\epsilon_0)\prod_{i=1}^{L_m}\rho_{e_i}(u_{e_i})}{\exp\{J_m(\xi_m)\}},
\label{eq:path-density}
\end{equation}
with
\[
 J_m(\xi_m)=J_r^F+\sum_{i=1}^{L_m}(J_{e_i}^{\mathcal B}+J_{e_i}^F)
 +J_m^{\mathrm{term}}+J_m^\chi.
\]
Here $\phi$ collects the learned transport parameters, including $\omega$
and a learned context's $\lambda$. The terms in $J_m$ are the log absolute
Jacobian determinants of the root flow, edge lifts and flows, terminal
flow, and support map, each evaluated at its input in \eqref{eq:path-map}.
Training evaluates the log of the density.

\MoveToAppendix{
\subsection{Derivation of the path density}
\label{app:path-regularity}
\paragraph{Regularity and support.}
Assume $p(m)>0$ and $0<Z_m<\infty$ for every $m\in\Mcal$.
The root and auxiliary variables are independent, with normalized densities
$b_0$ and $\rho_e$. Densities on $\Theta_m$ use $d_m$ independent coordinates,
omitting one simplex component when needed. Assume that the root,
dimension-matching, edge, terminal, and support maps and their inverses are locally
Lipschitz bijections on their open coordinate domains. Their Jacobians are
then nonsingular almost everywhere. For finite-valued objectives, also assume
$q_\phi(\cdot\mid m)\ll\pi_m$ and
\[
\E_{q_\phi(\cdot\mid m)}
\left|\log\frac{q_\phi(\theta_m\mid m)}{\eta_m(\theta_m)}\right|<\infty.
\]
\paragraph{Pullback representation.}
For the path $G_m=(r,e_1,\ldots,e_{L_m},v_m)$, define
\[
 \begin{aligned}
 \xi_m&=(\epsilon_0,u_{e_1},\ldots,u_{e_{L_m}}),\\
 r_m(\xi_m)&=b_0(\epsilon_0)\prod_{i=1}^{L_m}\rho_{e_i}(u_{e_i}).
 \end{aligned}
\]
Since $d_m=d_r+\sum_{i=1}^{L_m} r_{e_i}$, $\xi_m$ has dimension $d_m$.
Let $\Phi_{\phi,m}:\R^{d_m}\to\Theta_m$ denote the complete
path transformation in \eqref{eq:path-map}.  Expressing the target in
these path coordinates gives
\[
 \begin{aligned}
 J_{\phi,m}(\xi_m)&=\left|\det D\Phi_{\phi,m}(\xi_m)\right|,\\
 \widetilde\eta_{\phi,m}(\xi_m)&=\eta_m\{\Phi_{\phi,m}(\xi_m)\}
 J_{\phi,m}(\xi_m),\\
 \widetilde\pi_{\phi,m}&=\widetilde\eta_{\phi,m}/Z_m .
 \end{aligned}
\]
Under the bijectivity conditions below,
$\widetilde\pi_{\phi,m}$ is the \emph{pullback} of $\pi_m$ under
$\Phi_{\phi,m}$: it is the conditional posterior expressed in the path
coordinates $\xi_m$.  Accordingly, $\widetilde\eta_{\phi,m}$ has the same
normalizing constant $Z_m$ as $\eta_m$.

\begin{proposition}[Exact pullback to path coordinates]
\label{prop:path-pullback}
Suppose that the root flow, every dimension-matching map, every edge flow,
the terminal flow, and the support transformation on $G_m$ are locally bi-Lipschitz
bijections on the stated open coordinate domains, using independent
coordinates on constrained supports. The following identities hold almost
everywhere and for the indicated expectations.
\begin{equation}
\begin{aligned}
 q_\phi\{\Phi_{\phi,m}(\xi_m)\mid m\}
 &=\frac{r_m(\xi_m)}{J_{\phi,m}(\xi_m)},\\
 \E_{q_\phi(\cdot\mid m)}
 \left[\log\frac{q_\phi(\theta_m\mid m)}{\eta_m(\theta_m)}\right]
 &=\E_{r_m}
 \left[\log\frac{r_m(\xi_m)}
 {\widetilde\eta_{\phi,m}(\xi_m)}\right].
\end{aligned}
\label{eq:path-pullback}
\end{equation}
\end{proposition}

\begin{proof}
Write $\theta_m=\Phi_{\phi,m}(\xi_m)$. The change-of-variables theorem for
locally bi-Lipschitz maps gives $d\theta_m=J_{\phi,m}(\xi_m)d\xi_m$. Equating
the probability masses $q_\phi(\theta_m\mid m)d\theta_m$ and
$r_m(\xi_m)d\xi_m$ gives the first identity. Moreover,
$\widetilde\eta_{\phi,m}(\xi_m)=
\eta_m\{\Phi_{\phi,m}(\xi_m)\}J_{\phi,m}(\xi_m)$, so
\[
 \frac{q_\phi\{\Phi_{\phi,m}(\xi_m)\mid m\}}
      {\eta_m\{\Phi_{\phi,m}(\xi_m)\}}
 =\frac{r_m(\xi_m)}
       {\widetilde\eta_{\phi,m}(\xi_m)}.
\]
Changing variables in the expectation proves the second identity.  Finally,
$\int\widetilde\eta_{\phi,m}(\xi_m)d\xi_m=\int\eta_m(\theta_m)d\theta_m=Z_m$,
which yields the stated pullback density and the KL identity.
\end{proof}

The second identity is the conditional contribution to
\eqref{eq:joint-loss} written in path coordinates.  Dividing the two
densities by their common normalizing constant also gives
\[
\KL\{q_\phi(\cdot\mid m)\,\|\,\pi_m\}
=\KL\{r_m\,\|\,\widetilde\pi_{\phi,m}\}.
\]
Thus the conditional reverse-KL can be evaluated by drawing the base and
edge variables $\xi_m\sim r_m$ and transporting them along $G_m$.

\paragraph{Masked implementation.}
A fixed-width conditioner can receive a padded vector and an activity
mask. The padding is deterministic. Each affine coupling layer uses the
active outputs and sums the log scales over the native coordinates it
transforms.
}

\subsection{Tree-structured model distribution}
\label{sec:joint-weights}

We now construct the variational model distribution $q_\psi(m)$. At each
model state $v_m\in\mathcal V_{\mathrm M}$, we assign a locally normalized
categorical distribution $q_\psi(\cdot\mid v_m)$ over \texttt{stop} and the
available children. At an intermediate state, the draw selects a child.
Model $m$ is selected by following $G_m$ and
stopping at $v_m$, and thus
\begin{equation}
q_\psi(m)=q_\psi(\texttt{stop}\mid v_m)
\prod_{(v,c)\in G_m}q_\psi(c\mid v).
\label{eq:tree-model-distribution}
\end{equation}
At a leaf model state, $q_\psi(\texttt{stop}\mid v)=1$. At other states, we
assign a trainable logit to each available choice and apply a softmax over the
logits to obtain the categorical distribution. For a finite tree, this
parameterization represents every strictly positive distribution on $\Mcal$
(Appendix~\ref{app:tree-policy-representation}). The product of
\eqref{eq:tree-model-distribution} and $q_\phi(\theta_m\mid m)$ in
\eqref{eq:path-density} gives the complete variational density
$q_{\psi,\phi}(m,\theta_m)$.

The factorization in \eqref{eq:tree-model-distribution} lets us calculate
$q_\psi(m)$ along its path without evaluating the probabilities of all
candidate models. For variable selection with $p$ optional predictors, let
$\gamma_{1:j}\in\{0,1\}^j$ denote the inclusion indicators for the first
$j$ predictors; the full vector $\gamma\in\{0,1\}^p$ specifies a model.
Rather than separately parameterizing the conditional inclusion probability
for each partial vector $\gamma_{<j}$, we specify a shared logit function
$a_\psi(\gamma_{<j},j)$ and define the model distribution by

\begin{equation}
 q_\psi(\gamma)=
 \prod_{j=1}^p
 \operatorname{Bernoulli}\!\left[
 \gamma_j;\,\sigma\{a_\psi(\gamma_{<j},j)\}\right],
 \label{eq:implicit-sequential-policy}
\end{equation}
where $\sigma(\cdot)$ is the sigmoid function. The sequential design allows
us to sample and evaluate $q_\psi(\gamma)$ in $p$ steps without enumerating
the $2^p$ models. Setting $\gamma_j=1$ follows a birth edge and introduces a
new coefficient parameter, whereas $\gamma_j=0$ follows an identity edge
that retains the parent's parameters.

\subsection{Optimization}

We jointly optimize transport parameters $\phi$ and model-distribution
parameters $\psi$ with objective \eqref{eq:joint-loss}. Invertibility and
differentiability of the maps in \eqref{eq:path-map} admit reparameterization
gradients. We evaluate the discrete expectation over the model index by exact
enumeration or stochastic path sampling.

\MoveToAppendix{
\subsection{Exact-enumeration estimator}
For $S$ independent reparameterized draws from each conditional distribution, let
\[
 \widehat c_\phi(m)=\frac{1}{S}\sum_{s=1}^S
 \left[\log q_\phi\{\theta_m^{(s)}\mid m\}
       -\log\eta_m\{\theta_m^{(s)}\}\right].
\]
When the model set can be enumerated, we use
\begin{equation}
 \widehat{\mathcal L}_{\mathrm{enum}}
 =\sum_{m\in\Mcal}q_\psi(m)
 \left[\log q_\psi(m)-\log p(m)+\widehat c_\phi(m)\right].
 \label{eq:exact-enumerated-loss}
\end{equation}
Equation~\eqref{eq:tree-model-distribution} evaluates each $q_\psi(m)$
exactly, so automatic differentiation of this finite sum updates $q_\psi$
without sampling a discrete model.

\paragraph{Reparameterization gradients.}
\label{app:gradient-regularity}
Write a conditional draw as $\theta_m=\Phi_{\phi,m}(\xi_m)$ with a
parameter-independent base law, and define
\[
 f_m(\phi,\xi_m)=
 \log q_\phi\{\Phi_{\phi,m}(\xi_m)\mid m\}
 -\log\eta_m\{\Phi_{\phi,m}(\xi_m)\}.
\]
At each parameter value being differentiated, assume that $f_m$ is almost
surely differentiable and that, on a neighborhood $U$,
\[
 |f_m(\phi',\xi_m)-f_m(\phi'',\xi_m)|
 \leq M_m(\xi_m)\|\phi'-\phi''\|,
 \qquad \phi',\phi''\in U,\qquad \E M_m<\infty.
\]
This integrable local bound permits differentiation under the expectation,
so $\nabla_\phi\E f_m=\E\nabla_\phi f_m$. Together with the finite-loss
condition in Appendix~\ref{app:path-regularity}, it justifies unbiased
reparameterization gradients for both exact enumeration and path sampling.

\paragraph{Computational complexity.}
\label{app:complexity}
Write $G_m=(r,e_1,\ldots,e_{L_m},v_m)$ and let
$H=\max_{m\in\Mcal}L_m$ be the largest number of edges in a model path.
Let $C_\eta$, $C_0$, $C_\psi$, and $C_e$ denote worst-case
forward-and-reverse costs for, respectively, one target-density evaluation,
the root flow, optional terminal flow, and support map together, one local
policy evaluation, and one edge transport, including relevant log-Jacobians.
These costs may depend on the model dimension; the bounds use their maximum
over the model space.

\paragraph{Sequential path sampling.}
If a draw from $q_\psi$ is obtained by sampling its local decisions along the
path, one draw visits at most $H$ edges.  A batch of $B$ path-RLOO draws costs
\[
 T_{\mathrm{path,seq}}
 =O\!\left(B\{C_\eta+C_0+H(C_\psi+C_e)\}\right)
\]
and makes $B$ target evaluations.  The leave-one-out baseline is computed in
$O(B)$ time from the batch sum and does not alter this bound.  For the
sequential variable-selection policy, $H=p$, so a batch traverses $Bp$ local
decisions without forming a $2^p$-term objective.

\paragraph{Exact enumeration and sampling from explicit model probabilities.}
In \eqref{eq:exact-enumerated-loss}, the probability $q_\psi(m)$ is computed
once for each model, whereas its conditional contribution uses $S$ independent
draws.  Direct evaluation therefore has the bound
\[
 T_{\mathrm{enum}}
 =O\!\left(KH C_\psi
 +KS\{C_\eta+C_0+H C_e\}\right),
 \qquad K=|\Mcal|,
\]
and makes $KS$ target evaluations.  The first term evaluates $q_\psi(m)$ for
every model; the second draws and evaluates all model-conditional samples.
A finite, enumerable tree may also form all $K$ model probabilities before
sampling paths.  With direct path products, that choice gives
\[
 T_{\mathrm{path,term}}
 =O\!\left(KH C_\psi
 +B\{C_\eta+C_0+H C_e\}\right).
\]
It has the same $B$ target evaluations as sequential path sampling, but it
does not remove the $K$-dependent policy calculation.  The cardinality study
uses sequential sampling; the finite 15-, 128-, and 203-model implementations
form all model probabilities.

For reverse-mode differentiation, let $P_\psi+P_\phi$ be the parameter
storage and let $A_\eta$, $A_0$, $A_\psi$, and $A_e$ be the saved activation
storage for the target density, the root, optional terminal, and support
maps together, one local policy evaluation, and one edge transport.  Sequential path sampling requires
\[
 M_{\mathrm{path,seq}}
 =O\!\left(P_\psi+P_\phi+
 B\{A_\eta+A_0+H(A_\psi+A_e)\}\right)
\]
working memory, plus $O(|\mathcal V|+|\mathcal E|)$ if an explicit
construction tree is stored.  Forming a vector of all $K$ model probabilities
adds $O(K)$ values; direct path products can also retain $O(KH A_\psi)$ policy
graph storage.  For an all-at-once direct exact implementation, the
corresponding bound is
\[
 M_{\mathrm{enum,all}}
 =O\!\left(P_\psi+P_\phi+|\mathcal V|+|\mathcal E|
 +KH A_\psi+KS\{A_\eta+A_0+H A_e\}\right).
\]
Processing $R<K$ model blocks at a time can replace $K$ by $R$ in the
conditional activation term, at the cost of recomputing or separately
accumulating the associated policy contributions.  A shared sequential policy
avoids a separate model-specific logit or explicit tree for every subset.

These are algorithmic operation counts and do not imply a universal ordering
of actual computing time across variational families: the architecture- and
target-specific values of $C_\eta$, $C_\psi$, and $C_e$ must be measured in a
matched implementation.
}

\paragraph{Stochastic path-RLOO estimator.}
We draw independent model states $m_b\sim q_\psi$, for $b=1,\dots,B$ with
$B\geq2$, and generate $\theta_b$ using the path transport in
\eqref{eq:path-map} along $G_{m_b}$. The Monte Carlo estimator
$\widehat{\mathcal L}_{\mathrm{path}}=B^{-1}\sum_{b=1}^B\ell_b$ is
unbiased for \eqref{eq:joint-loss}. Write $\ell_b=\ell(m_b,\theta_b)$ with
\begin{equation}
\begin{aligned}
\ell(m,\theta_m)&=\log q_\psi(m)+\log q_\phi(\theta_m\mid m)\\
&\quad-\log p(m)-\log\eta_m(\theta_m).
\end{aligned}
\label{eq:path-sampled-loss-term}
\end{equation}
For estimating the gradients, we use reparameterization for the continuous
parameters and a REINFORCE score-function estimator with a leave-one-out
control variate for the model states
\citep{williams1992,kool2019baseline,richter2020vargrad}:
\begin{equation}
 \widehat{\nabla_\psi\mathcal L}_{\mathrm{path}}
 =\frac{1}{B}\sum_{b=1}^B
 \operatorname{sg}\!\left(\ell_b-\bar\ell_{-b}\right)
 \nabla_\psi\log q_\psi(m_b),
 \label{eq:path-decision-gradient}
\end{equation}
Here $\bar\ell_{-b}=(B-1)^{-1}\sum_{j\ne b}\ell_j$ and
$\operatorname{sg}$ is the stop-gradient operator, which acts as the identity
in the forward pass but has zero derivative. The additional unit score term
arising from $q_\psi(m)\log q_\psi(m)$ has zero expectation and may be
retained or omitted. We call the estimator in
\eqref{eq:path-decision-gradient} \emph{path-RLOO}.
Appendix~\ref{app:tree-policy-gradient} proves its unbiasedness under
standard integrability conditions.

\paragraph{Sampling from the fitted approximation.}
After fitting, draw $m\sim q_\psi$ and then, conditional on $m$, draw
$\epsilon_0\sim b_0$ and $u_e\sim\rho_e$ for every $e\in G_m$.  Applying
$\Phi_{\phi,m}$ to these variables produces $\theta_m$; its joint density
is evaluated from \eqref{eq:tree-model-distribution} and
\eqref{eq:path-density}.

\paragraph{Computational cost.}
Suppose $q_\psi$ is sampled via the local distributions in
\eqref{eq:tree-model-distribution}. Let $H$ be the largest number of edges
in a model path, and let $C_\eta$, $C_0$, $C_\psi$, and $C_e$ denote the
forward-and-reverse costs of one target-density evaluation, the root,
optional terminal, and support maps together, one local policy evaluation,
and one edge transport, including relevant log-Jacobians. For $B$ sampled paths, path-RLOO has
per-iteration cost
\begin{equation}
 T_{\mathrm{path,seq}}
 =O\!\left(B\{C_\eta+C_0+H(C_\psi+C_e)\}\right).
\label{eq:cost}
\end{equation}
Exact enumeration computes $K=|\Mcal|$ model probabilities and performs
$S$ conditional draws per model. Forming all terminal probabilities before
sampling paths adds a $K$-dependent policy cost. Appendix~\ref{app:complexity}
gives the corresponding time and memory bounds, and
Appendix~\ref{app:runtime} reports measured update costs.

\subsection{Model-mass tilt under joint reverse KL}
\label{sec:joint-calibration}

\paragraph{Conditional-gap tilt.}
For the conditional density $q_\phi(\cdot\mid m)$ of SM-VTI, if there is no
constraint on the model distribution, the optimal model mass obtained by
minimizing the reverse-KL is the same as that derived for VTI by
\citet{davies2025}, that is,
\begin{equation}
q_\psi^\star(m)=
\frac{p(m)Z_m\exp\{-\delta_m(\phi)\}}
{\sum_{j\in\Mcal}p(j)Z_j\exp\{-\delta_j(\phi)\}},
\label{eq:conditional-gap-tilt}
\end{equation}
where $\delta_m(\phi)=\KL\{q_\phi(\cdot\mid m)\,\|\,\pi_m\}$ is the reverse
K-L divergence between $q_\phi(\cdot\mid m)$ and the conditional target
$\pi_m(\theta_m)=\eta_m(\theta_m)/Z_m$, and
$Z_m=\int_{\Theta_m}\eta_m(\theta_m)\,d\theta_m$ is the model evidence.
See Appendix~\ref{app:conditional-gap-tilt} for the derivation.

The result above implies two things. First, unless the conditional-approximation
gaps are equal across models, the optimal model mass $q_\psi^\star(m)$ may not
equal the target $p(m\mid\mathbf y)$. In fact, models with conditional
posteriors that are easier to approximate can receive more mass.
Second, for a finite model set, the fitted conditional density provides an
importance proposal for the corresponding model evidence~\citep{mueller2019neural}:
\begin{equation}
\widehat\pi_M^{\mathrm{IS}}(m)=
\frac{p(m)\widehat Z_m}
{\sum_{j\in\Mcal}p(j)\widehat Z_j},
\label{eq:importance-model-marginal}
\end{equation}
where $\widehat Z_m=\frac{1}{S}\sum_{s=1}^S
\frac{\eta_m\{\theta_m^{(s)}\}}
{q_\phi\{\theta_m^{(s)}\mid m\}}$, with
$\theta_m^{(s)}\sim q_\phi(\cdot\mid m)$.

\MoveToAppendix{
\subsection{Derivation of the conditional-gap tilt}
\label{app:conditional-gap-tilt}

\begin{proposition}
\label{prop:conditional-gap-tilt}
Let $\Mcal$ be finite, suppose $0<p(m)Z_m<\infty$ for every $m$, and fix
conditional densities for which
$\delta_m=\KL\{q_\phi(\cdot\mid m)\,\|\,\pi_m\}<\infty$.
If $q_\psi$ ranges over the interior of the model-probability simplex, the
unique minimizer of the joint reverse KL over $q_\psi$ is
\eqref{eq:conditional-gap-tilt}.
\end{proposition}

\begin{proof}
Let $Z_{\mathrm{tot}}=\sum_{j\in\Mcal}p(j)Z_j$.  Since
$p(m,\theta_m\mid\mathbf y)=p(m)Z_m\pi_m(\theta_m)/Z_{\mathrm{tot}}$, the joint reverse KL equals
\[
 \log Z_{\mathrm{tot}}+\sum_{m\in\Mcal}q_\psi(m)
 \left[\log q_\psi(m)-\log\{p(m)Z_m\}+\delta_m\right].
\]
Set $a_m=p(m)Z_m\exp\{-\delta_m\}$ and $A=\sum_m a_m$.  The displayed
quantity is $\log Z_{\mathrm{tot}}-\log A+\KL(q_\psi\,\|\,a/A)$, which is uniquely
minimized by $q_\psi(m)=a_m/A$.
\end{proof}

\paragraph{Importance-proposal support.}
For the evidence estimator in \eqref{eq:importance-model-marginal}, assume
$\pi_m\ll q_\phi(\cdot\mid m)$ and independent proposal draws conditional
on the fitted parameters. This support condition ensures that every region
of positive posterior mass can be sampled. It is the reverse of the support
condition for the variational objective. No finite-variance assumption is
needed to define the estimator; accurate finite-sample estimation still
depends on the proposal's tail coverage.
}

\section{Experiments}
\label{sec:experiments}

We present studies including an exact benchmark, robust variable selection, nonlinear DAG
discovery on synthetic and flow-cytometry data, and
nucleotide substitution-model selection in phylogenetics. For variable selection and substitution-model selection, we compare methods using the total variation (TV) of model
probabilities and the
approximation of the joint negative log likelihood (NLL) against references generated by RJMCMC. The method performs better if the two measures become lower. For DAG discovery, we compare samples of the
posterior graph with the generating graph in simulation or the
reference graph for real data. Appendix~\ref{sec:evaluation}
gives the formal definition of the measures. The accompanying anonymized
code bundle contains the experiment configurations and evaluation commands.
For the released VTI factory, $\operatorname{Affine}(a,b)$ denotes $a$
masked affine-autoregressive stages with $b$ residual MADE blocks per
conditioner. Width and context settings are specified separately.

\MoveToAppendix{
\section{Experimental settings and evaluation}
\subsection{Evaluation measures}
\label{sec:evaluation}

For the variable-selection and phylogenetic studies, let
$\{(m^{(s)},\theta_{m^{(s)}}^{(s)})\}_{s=1}^S$ be a collection of
RJMCMC posterior states that is not used in fitting, and let
$\widehat\pi_M$ denote the empirical model marginal computed from the retained
reference chains.  The states used for density evaluation may be subsampled
from these chains.  Our primary joint measure is
the negative log density
\begin{equation}
 \widehat{\mathrm{NLL}}_{\mathrm{joint}}
 =-\frac{1}{S}\sum_{s=1}^S
 \left[\log q_\psi(m^{(s)})+
 \log q_\phi(\theta_{m^{(s)}}^{(s)}\mid m^{(s)})\right].
\label{eq:joint-nll}
\end{equation}
This is a Monte Carlo estimate of the joint cross-entropy
$H\{p(\cdot,\cdot\mid\mathbf y),q_{\psi,\phi}\}$; chain dependence affects its Monte Carlo uncertainty,
not the quantity being estimated.  The aggregate conditional NLL is
$-S^{-1}\sum_{s=1}^S\log q_\phi(\theta_{m^{(s)}}^{(s)}\mid m^{(s)})$.
For a model distribution $\nu_M$, we report total variation distance
\[
 \TV(\nu_M,\widehat\pi_M)
 =\frac12\sum_{m\in\Mcal}
 \left|\nu_M(m)-\widehat\pi_M(m)\right|.
\]
Here $\nu_M=q_\psi$ for the directly optimized approximation and, where
indicated, $\nu_M=\widehat\pi_M^{\mathrm{IS}}$ for the importance-sampling
correction in \eqref{eq:importance-model-marginal}.  A per-model conditional NLL is the Monte Carlo estimate of
$H\{\pi_m,q_\phi(\cdot\mid m)\}$ from reference draws conditional on $m$.
Top-$k$ overlap is the number of models shared by the learned and reference
top-$k$ sets.
}

\subsection{Exact tree-structured posterior}
\label{sec:hwgt-joint}

\MoveToAppendix{
\subsection{Controlled tree target}
We define a controlled target directly through a known normalized joint
posterior.  Its model space is a complete depth-three binary tree in which
every node represents a model.
At an internal node $v$, fixed probabilities
$\alpha_v(\mathrm{stop}),\alpha_v(L),\alpha_v(R)$ define the mass of a
model $m$ as the product of the child probabilities on its root-to-$m$
path and the final \texttt{stop} probability,
\[
\pi_M(m)=\prod_{(v\to c)\in\operatorname{path}(m)}\alpha_v(c)\,
 \alpha_m(\mathrm{stop}).
\]
The fixed local probabilities assign posterior mass to models at multiple tree
depths.  A model at depth $r$ has parameter
\[
 \theta_m=(\theta_{m,1},\theta_{m,2},\beta_{m,1},\ldots,\beta_{m,r})
 \in\R^{2+r}.
\]
The first two coordinates are retained along every path, and each edge adds
one coordinate.  Conditional on $m$, the target first draws a correlated Gaussian
$z\sim\mathcal N_{d_m}(\mu_m,L_mL_m^\top)$ and then applies the triangular
unit-Jacobian map, using zero-based coordinate indices,
\begin{align*}
 \theta_0&=z_0,\\
 \theta_j&=z_j+0.55\sin(z_{j-1})+0.20s_{m,j}\tanh(z_0)\\
 &\quad+0.08\cos((j+1)z_0),\qquad 1\leq j<d_m,
\end{align*}
where $s_{m,1}=1$ and, for $j\geq2$, $s_{m,j}$ is $-1$ or $1$ according
to whether the path to $m$ takes a left or right branch at depth $j-1$.
Calling the resulting
normalized density $f_m(\theta_m)$ and taking $p(m)=1/15$, we set
\[
 \eta_m(\theta_m)=f_m(\theta_m)\frac{\pi_M(m)}{p(m)}.
\]
Consequently $p(m)\eta_m(\theta_m)=\pi_M(m)f_m(\theta_m)$ is already the
exact normalized joint posterior.  The benchmark therefore evaluates recovery
of both the tree-structured model distribution and the model-specific
conditional distributions.
}

We first assess joint recovery of SM-VTI on a controlled 15-model tree with the posterior in closed form.  Across five independent runs of 1,500 training
iterations, the average model TV is $0.03155\pm0.00044$ (mean $\pm$ sample SD) and joint NLL is
$8.2287\pm0.0244$. Figure~\ref{fig:hwgt} compares model probabilities and
the geometry of conditional samples for one deep model in a representative
fit. At fixed flows, the estimated optimal model masses have mean TV
$0.00971$ to the exact posterior and reduce joint reverse KL by a mean
$0.00256$ (Appendix~\ref{app:fixed-flow-audit}).

\MoveToAppendix{
\subsection{Robust variable-selection setting}
\label{app:experiments}
}
\begin{figure*}[!t]
\centering
\includegraphics[width=\textwidth]{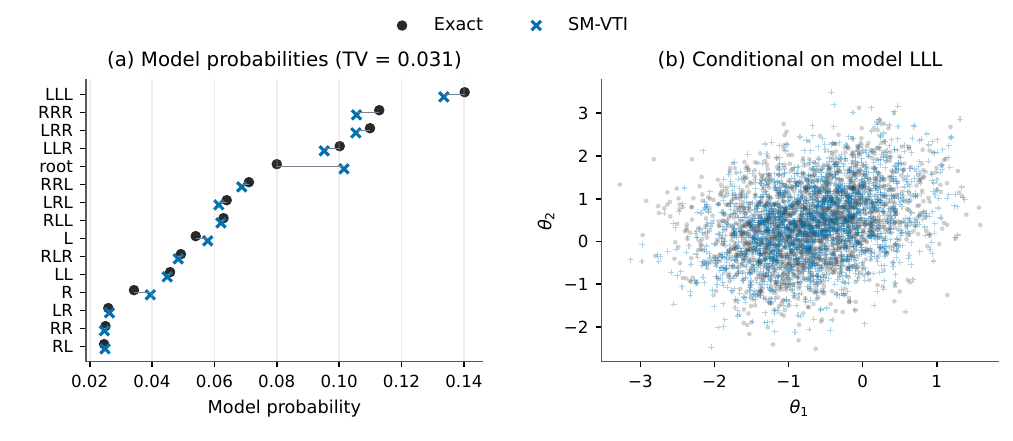}
\caption{\textbf{Exact tree-structured posterior benchmark.}
(a) Exact and fitted probabilities at all 15 model states after 1,500
updates. L and R identify branches from the root. (b) A two-dimensional
projection of 2,000 draws each from the exact and fitted conditionals for
model LLL. Both panels use the same fitted checkpoint.}
\label{fig:hwgt}
\end{figure*}

\begin{figure*}[!t]
\centering
\includegraphics[width=\textwidth]{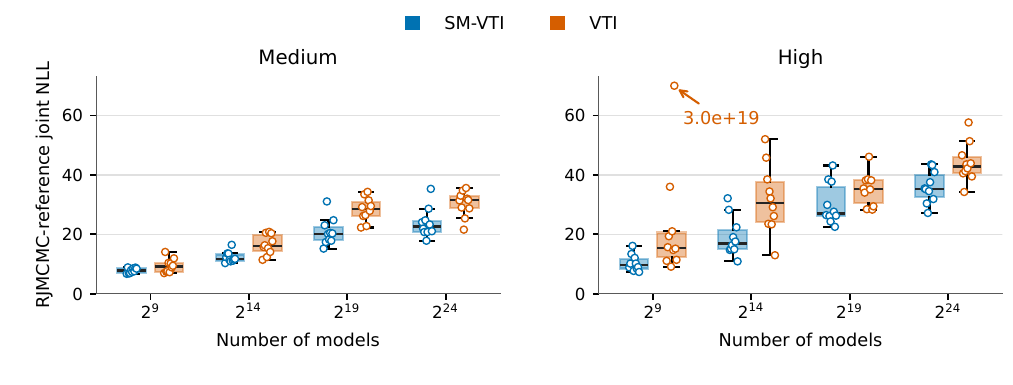}
\caption{\textbf{Cardinality extension in robust variable selection.}
Joint RJMCMC-reference NLL of $2^9$--$2^{24}$ models
for \method{} and \vtibaseline{} under Medium and High data
respectively. The results for ten paired runs under each setting are shown by
the points and box plots, including those flagged by the tail diagnostic
detailed in Appendix~\ref{app:experiments}. VTI uses a MADE-model distribution
throughout. The extreme High value for VTI ($3.0\times10^{19}$) is
indicated at the display limit and remains in the statistics.}
\label{fig:vs-cardinality-affine}
\end{figure*}

\subsection{Bayesian robust variable selection}
\label{sec:vs128-joint}

Following \citet{davies2025}, we study robust variable selection with $7$ optional predictors.  For
$\gamma\in\{0,1\}^7$, the mean is modeled as $\mu_i=\beta_0+x_i^\top(\gamma\odot\beta)$.
The likelihood has the identical form to the released VTI implementation, with
the observation $y_i= \mu_i + r_i$ and the residual factor
$0.9\exp(-r_i^2/2)+0.1\exp(-r_i^2/200)$.  We use a uniform prior on models and
Gaussian priors on active coefficients.  `Medium' and `High' are misspecified
data generators; High additionally correlates some inactive predictors with
signal predictors compared to Medium. Appendix~\ref{app:experiments} gives
the data generators and training configurations.

\MoveToAppendix{
\paragraph{Analysis model and data generation.}
We use the 128-model robust variable-selection setting of \citet{davies2025}.
An intercept is always active and seven scalar
predictors are optional, so $m=\gamma\in\{0,1\}^7$ and
$d_m=1+\lVert\gamma\rVert_0$.  With
$\mu(x)=\beta_0+\sum_{j=1}^7x_j\gamma_j\beta_j$ and
$r_i=y_i-\mu(x_i)$, the inference likelihood uses the residual kernel
\[
 \ell(r_i)=0.9\exp(-r_i^2/2)+0.1\exp(-r_i^2/200).
\]
We follow the released VTI implementation; after normalization, its
mixture weights are $9/19$ and $10/19$.
The prior is
\[
 \begin{aligned}
 p(\beta_0,\beta_\gamma\mid\gamma)
 &=\prod_{j\in\{0\}\cup\{k:\gamma_k=1\}}
 \mathcal N(\beta_j;0,1.5^2),\\
 p(\gamma)&=2^{-7}.
 \end{aligned}
\]

Both methods target this analysis posterior.  Each data set has 50
observations.  The intercept and two randomly selected predictors have nonzero
generating coefficients, equal to $0.5$ for the first 25 observations and
$1.5$ for the remaining 25.  All responses receive Gaussian base noise; five randomly
selected responses also receive independent Gaussian noise.  The base and
additional noise standard deviations are $(2,5)$ for Medium and $(4,4)$
for High.  High also correlates active and inactive covariates.

\paragraph{Tree representation for subset selection.}
The \method{} tree uses a fixed predictor order.  A node indexed by the included set
$S\subseteq\{1,\ldots,7\}$ has one \texttt{stop} option and one child
$S\cup\{j\}$ for each $j>\max S$ (with $\max\varnothing=0$).  Every subset
therefore has exactly one increasing-index path.  For a child that adds
predictor $j$, model-specific coordinates are ordered as the intercept
followed by the selected coefficients in increasing index.  The Gaussian
dimension-matching map is
\[
 \begin{aligned}
 \mathcal B_{S,j}(z_S,u_j)&=
 \left(z_S,\;\mu_\phi(z_S,h_{S,j})+
 e^{s_\phi(z_S,h_{S,j})}u_j\right),\\
 u_j&\sim\mathcal N(0,1),\\
 \log|\det D\mathcal B_{S,j}|&=s_\phi(z_S,h_{S,j}).
 \end{aligned}
\]
Here $h_{S,j}$ encodes the parent model and added predictor.  The map
identifies the newly added coefficient, and each birth is followed by two
shared affine coupling layers.
}

\MoveToAppendix{
\paragraph{VTI comparison.}
The VTI comparison uses categorical model probabilities, its score-function
estimator, and a two-layer affine CoSMIC transport.  Both methods use float64
arithmetic, 30,000 iterations, and 1,024 target evaluations per iteration.
The paired runs use the same simulated data sets.  Thus the comparison matches
the target-evaluation budget and closely matches the number of learned
parameters, while retaining the distinct variational constructions described
in Sections~\ref{sec:background-avti} and~\ref{sec:method}.
}

\MoveToAppendix{
\paragraph{Stochastic path sampling.}
On the same ten data sets per setting, both estimators use the same
\method{} family and target-evaluation budget. Path-RLOO replaces the
finite sum with independent paths and uses \eqref{eq:path-decision-gradient}.
At 30,000 updates, median model TV is $0.274$ with path-RLOO and $0.276$
with exact enumeration on Medium and $0.160$ for both on High. Median
joint NLL is $5.625$ and $5.687$ on Medium and $6.872$ and $6.963$ on High,
respectively.
}

\begin{table}[H]
\caption{\textbf{Comparison of \method{} and \vtibaseline{} using nearly parameter-matched affine.}
Ten paired data sets are analysed. Model TV is reported as mean (SD).
Joint NLL is reported as mean (SD) for Medium and as median for High
because one VTI run is extreme; first and third quartiles are in
Table~\ref{tab:vs128-affine55-capacity}.}
\label{tab:vs128-main-results}
\centering
\scriptsize
\begin{tabular}{@{}llrr@{}}
\toprule
\textbf{Setting} & \textbf{Method} & \textbf{Model TV} & \textbf{Joint NLL} \\
\midrule
Medium & \method{} & \best{0.252\,(0.087)} & \best{5.665\,(1.095)} \\
       & \vtibaseline{} & $0.260\,(0.077)$ & $6.209\,(1.692)$ \\
High   & \method{} & \best{0.159\,(0.056)} & \best{6.963} \\
       & \vtibaseline{} & $0.225\,(0.099)$ & $10.475$ \\
\bottomrule
\end{tabular}
\end{table}

Table~\ref{tab:vs128-main-results} matches the likelihood-evaluation budget
and approximately matches the parameter count. This comparison uses exact
summation over models with Monte Carlo draws for each SM conditional and a
categorical score estimator for VTI. The within-SM estimator control in
Appendix~\ref{app:experiments} holds the family fixed. The primary
comparison shows comparable model TV and lower
joint NLL of \method{} compared to VTI of \citet{davies2025} for the Medium data, and lower values on both measures for \method{} under the High
setting.  VTI equipped with the $\mathrm{Affine}(5,5)$ transport produces 
comparable Medium TV and joint NLL to \method{}, while \method{} still has
lower model TV for High data (Table~\ref{tab:vs128-affine55-capacity}).
Under common Adam settings and exact model summation, \method{} has lower
joint NLL than near-equal-size CoSMIC on six of ten Medium data sets.
Wider CoSMIC narrows the median paired joint NLL gap from 0.161 to 0.009
(Appendix~\ref{app:comparison-protocols}). Against parameter-matched CoSMIC
given the full construction path, \method{} lowers conditional and
joint NLL in eight of ten Medium and all three High pairs
(Table~\ref{tab:vs-path-matched-capacity}).

Appendices~\ref{app:order-sensitivity} and~\ref{app:topology-sensitivity}
compare predictor orders and two subset trees with different branching under
matched training. From ascending to alternating, median paired joint NLL
changes are $-0.0024$ on Medium and $+0.0338$ on High, with larger changes
on some individual High data sets.

\MoveToAppendix{
\begin{figure}[!htbp]
\centering
\includegraphics[width=.92\textwidth]{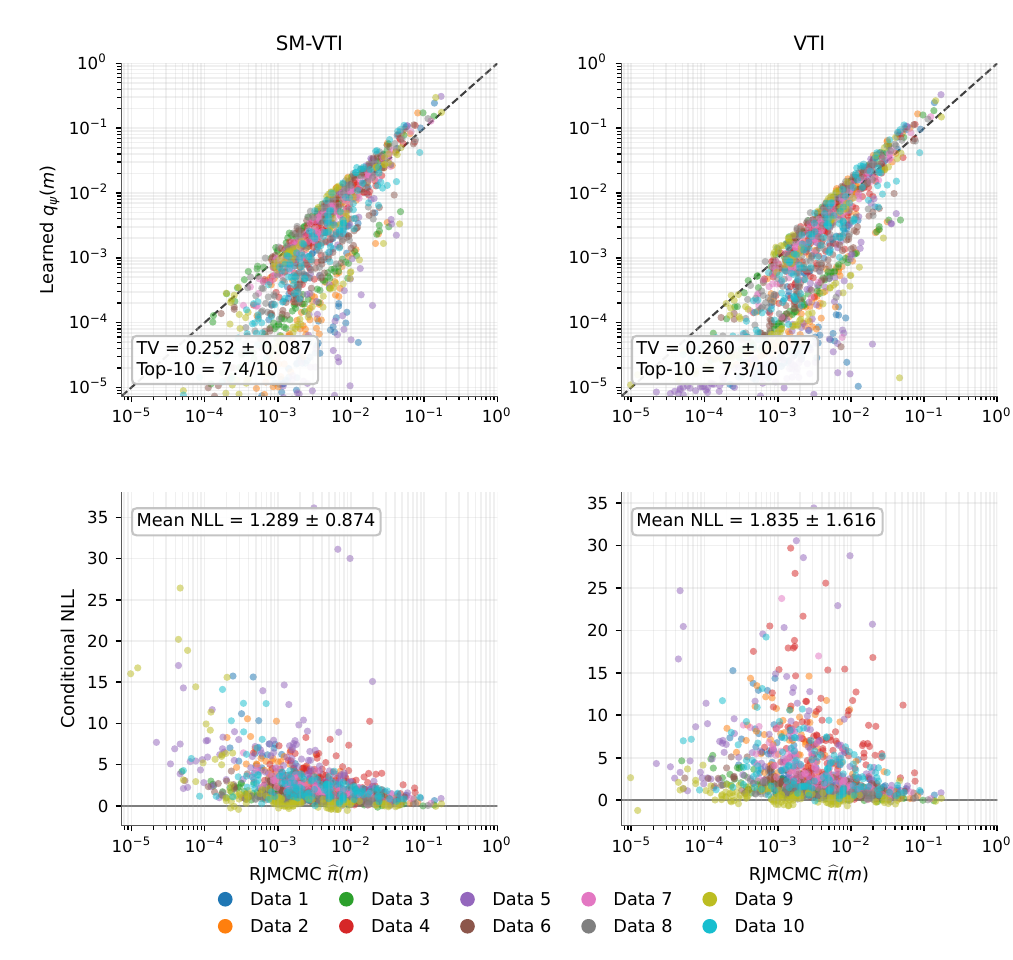}
\caption{\textbf{Model probabilities and conditional cross-entropies under
Medium misspecification at 30,000 iterations.} The upper row compares learned
model probabilities with RJMCMC frequencies; the lower row compares
model-conditional cross-entropies.  Colours identify data sets and the dashed
line denotes equality.}
\label{fig:vs128-affine}
\end{figure}
}

\MoveToAppendix{
\begin{table*}[t]
\caption{\textbf{VTI $\mathrm{Affine}(5,5)$ capacity sensitivity at 30,000
iterations.} The comparison replaces the two-layer $\mathrm{Affine}(2,2)$
transport in Table~\ref{tab:vs128-main-results} by
$\mathrm{Affine}(5,5)$.  Entries are mean (sample SD) over ten paired runs,
except for High joint NLL, reported as median [first, third quartiles].}
\label{tab:vs128-affine55-capacity}
\centering
\small
\begin{tabular}{@{}llrrr@{}}
\toprule
\textbf{Setting} & \textbf{Statistic} & \textbf{\method{}} & \textbf{\boldmath VTI $\mathrm{Affine}(2,2)$} & \textbf{\boldmath VTI $\mathrm{Affine}(5,5)$} \\
\midrule
Medium & Model TV & \best{0.252\,(0.087)} & $0.260\,(0.077)$ & \best{0.252\,(0.081)} \\
       & Joint NLL & $5.665\,(1.095)$ & $6.209\,(1.692)$ & \best{5.615\,(1.000)} \\
\addlinespace
High   & Model TV & \best{0.159\,(0.056)} & $0.225\,(0.099)$ & $0.192\,(0.067)$ \\
       & Joint NLL & \best{6.963\,[6.577,8.534]} & $10.475\,[8.073,21.503]$
         & $9.393\,[6.690,1.24\times10^{11}]$ \\
\bottomrule
\end{tabular}
\end{table*}
}

\MoveToAppendix{
\begin{figure}[!htbp]
\centering
\includegraphics[width=.92\textwidth]{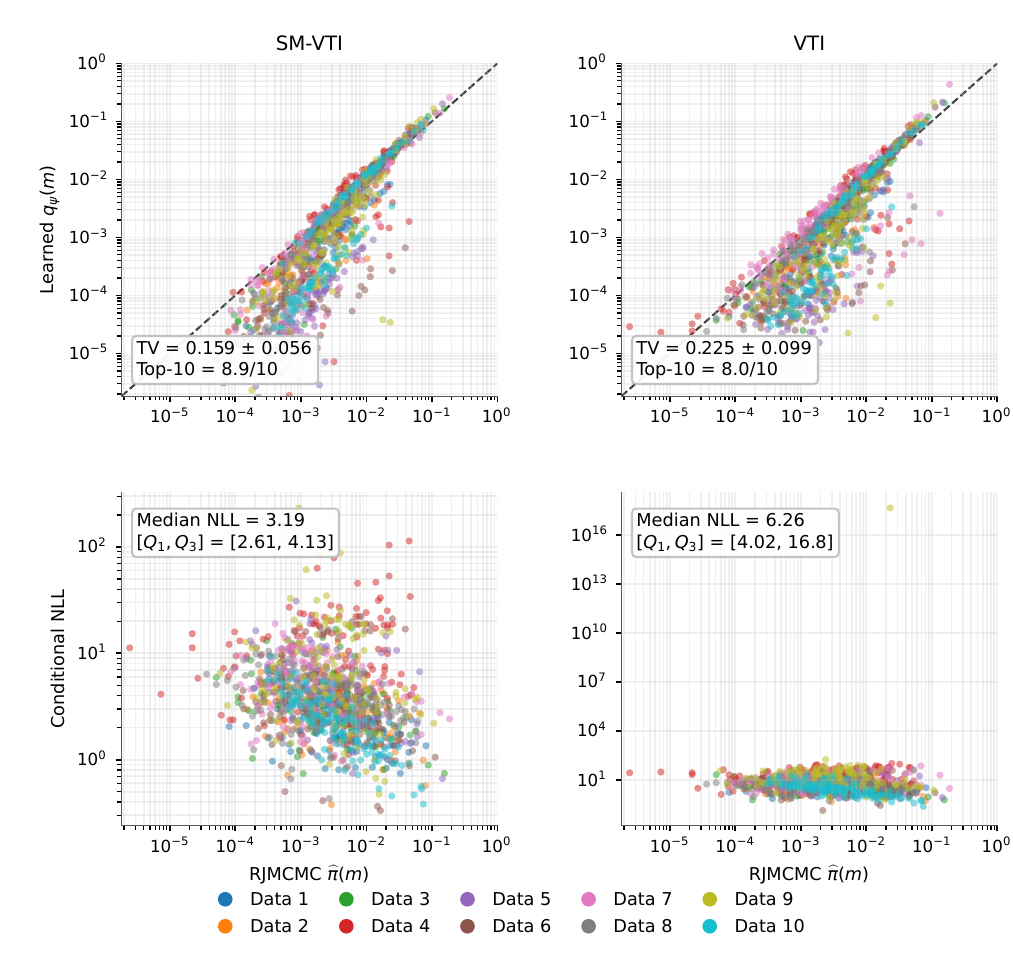}
\caption{\textbf{Model probabilities and conditional cross-entropies under
High misspecification at 30,000 iterations.} The upper row compares learned
model probabilities with RJMCMC frequencies; the lower row compares
model-conditional cross-entropies on a logarithmic scale.  Colours identify
data sets and the dashed line denotes equality.}
\label{fig:vs128-high-affine}
\end{figure}
}

\paragraph{Cardinality sweep.}
To assess posterior approximation of \method{} in larger model spaces, we vary the number
of optional predictors over $p\in\{9,14,19,24\}$, giving $|\Mcal|=2^p$
candidate models. Figure~\ref{fig:vs-cardinality-affine} shows the joint
NLL against the reference generated with RJMCMC posterior draws for ten paired data sets at each $p$
under Medium and High misspecification. \method{} generally has lower
joint NLL than VTI. At $p=24$, the medians are
$22.6$ and $31.5$ under Medium misspecification, $35.3$ and $42.9$
under High, for \method{} and VTI, respectively.
Figure~\ref{fig:vs-bivariate} in Appendix~\ref{app:experiments} shows a
128-model Medium example. Its coefficient marginals include a spike at zero
for the optional coefficient. Appendix~\ref{app:vs-quality-time} reports
quality against recorded training time for $p=24$.

\MoveToAppendix{
\paragraph{Cardinality-extension protocol.}
The sweep extends the robust analysis likelihood and Gaussian
coefficient-prior form above to $p$ optional predictors.
Each cardinality and misspecification level has ten paired synthetic data
sets.  Both methods train for 30,000 iterations with 1,024 target evaluations
per iteration.  \method{} samples binary inclusion decisions using
\eqref{eq:implicit-sequential-policy}; selected coefficients enter through
typed Gaussian birth maps, followed by two shared affine coupling layers
applied once in the completed native chart.  VTI uses a MADE model
distribution with a two-layer affine CoSMIC transport.  The methods share an
update and target-evaluation budget, and their learned parameter counts differ
by less than 1\% at each cardinality.  The fixed 128-model comparison in
Table~\ref{tab:vs128-main-results} instead uses categorical VTI logits and a
per-birth shared \method{} flow.

We evaluate held-out joint NLL on RJMCMC posterior draws because the
large-space fits do not enumerate terminal model probabilities for a model-TV
comparison.  The target and parameter dimension change with $p$, so raw NLL
differences across sizes combine changes in target entropy and approximation
error; we interpret comparisons at fixed $p$ and misspecification level.
All 80 paired runs are retained. For each evaluated conditional define
$\ell_{\mathrm{cond}}=-\log q_\phi(\theta_m\mid m)$ on the finite held-out
states. The tail flag records whether the maximum reaches $10^3$.
Nine pairs have a flag at one or more checkpoints. At the plotted endpoint,
four VTI fits and no \method{} fits reach the threshold. The finite tail
flag and the forward/inverse density audit measure different properties.
The extreme High $p=9$ VTI endpoint is retained and annotated in
Figure~\ref{fig:vs-cardinality-affine}.

\begin{figure}[!htbp]
\centering
\includegraphics[width=.92\textwidth]{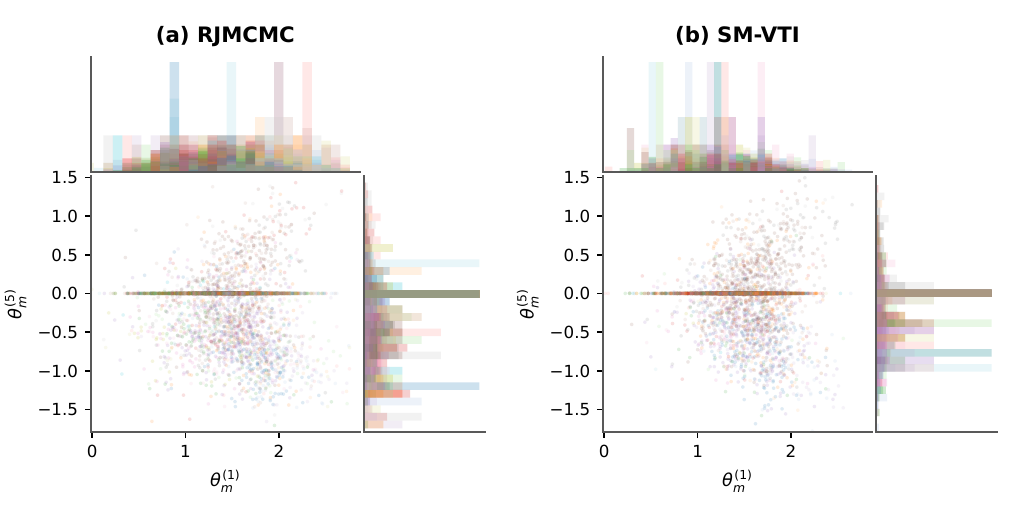}
\caption{\textbf{Selected coefficient projection for a 128-model Medium data
set.} (a) Reference and (b) fitted draws at 30,000 updates. Colour indicates
model state. Inactive coefficients are plotted at zero.}
\label{fig:vs-bivariate}
\end{figure}
}

\subsection{Bayesian nonlinear directed acyclic graph discovery}
\label{sec:dag-synthetic}

\begin{table*}[!t]
\caption{\textbf{Graph recovery on small-sample nonlinear DAGs.}
Entries average one fit per method and data set over three paired data sets per sample size.
\method{} and VTI integrate the Gaussian output coefficients exactly
and use the same training recipe for 25,000 updates.
All methods use the same observations and reference graphs.
The four additional baselines use their own models and training budgets.}
\label{tab:dag-small-integrated}
\centering
\small
\begin{tabular}{@{}lrrrrrrrr@{}}
\toprule
 & \multicolumn{4}{c}{$N=16$} & \multicolumn{4}{c}{$N=32$} \\
\cmidrule(lr){2-5}\cmidrule(lr){6-9}
Method & F1 $\uparrow$ & SHD $\downarrow$ & Brier $\downarrow$ & AUROC $\uparrow$
       & F1 $\uparrow$ & SHD $\downarrow$ & Brier $\downarrow$ & AUROC $\uparrow$ \\
\midrule
\method{} (integrated) & 0.454 & \textbf{22.89} & 14.79 & 0.642 & 0.537 & \textbf{18.54} & 15.14 & 0.687 \\
VTI (integrated) & 0.439 & 24.77 & \textbf{14.24} & 0.631 & 0.498 & 21.18 & \textbf{13.64} & 0.665 \\
DAGMA & \textbf{0.474} & 34.00 & 34.00 & \textbf{0.643} & \textbf{0.612} & 24.00 & 24.00 & \textbf{0.756} \\
DiBS & 0.155 & 24.97 & 22.13 & 0.533 & 0.223 & 24.33 & 20.49 & 0.553 \\
DiBS+ & 0.220 & 23.46 & 22.60 & 0.556 & 0.364 & 20.66 & 20.18 & 0.610 \\
JSP-GFN & 0.372 & 27.56 & 21.93 & 0.609 & 0.348 & 32.12 & 29.02 & 0.598 \\
\bottomrule
\end{tabular}
\end{table*}

We infer a directed acyclic graph (DAG) jointly with the parameters of a
nonlinear structural equation model,
$X_j=f_j(X_{\mathrm{pa}(j)};\theta_j)+\varepsilon_j$, where
$\mathrm{pa}(j)$ denotes the parents of node $j$ and $\varepsilon_j$ is
Gaussian noise. The conditional means are parameterized by neural networks
whose active parameters depend on the parent sets, yielding a
transdimensional joint posterior. We represent the graph as $P^\top UP$,
with a permutation matrix $P$ and a strictly upper-triangular binary $U$,
using the ordering representation of \citet{davies2025}, related to
\citet{cundy2021bcdnets}.

We extend the VTI comparison with nonlinear
DAGMA~\citep{bello2022}, DiBS/DiBS+~\citep{lorch2021}, and
JSP-GFN~\citep{deleu2023}. The baselines share observations and reference
graphs and use their own models and training budgets. \method{} and VTI
share a Bayesian target and target-evaluation budget.

\paragraph{Synthetic data.}
Following \citet{davies2025}, we study ten-node DAGs with
$N\in\{16,32,64,128,256,512,1024\}$ observations.
At $N=16,32$, both VI methods integrate the Gaussian output coefficients
conditional on the graph and hidden weights. The integral preserves their
shared graph posterior. In Table~\ref{tab:dag-small-integrated}, DAGMA has
the highest mean F1 and AUROC at both sizes. The lowest mean SHD comes
from \method{}, and the lowest summed edge Brier from VTI.
DAGMA recovers more true edges and also predicts more false edges.
When each posterior method returns the DAG maximizing its fitted expected
F1, mean reference F1 at $N=16$ is $0.593$ for \method{}, $0.610$ for VTI,
and $0.474$ for DAGMA. At $N=32$, DAGMA leads on F1 under that comparison
(Appendix~\ref{app:dag-point-decisions}).
Integration improves F1 and negative ELBO in both families at all saved
checkpoints through 25,000 updates. This pattern recurs at $N=16$ with two
additional optimizer starts per data set (Appendix~\ref{app:dag-output-integral}).
On four nodes, both families also have lower graph TV to the numerical
reference and closer agreement with reference credible-set masses
(Table~\ref{tab:dag-graph-reference}).
At $N\geq256$, the original
\method{} fits have better means than VTI on all four scores.
Table~\ref{tab:dag-additional-baselines} gives the full sample-size comparison.

\paragraph{Flow-cytometry data.}
Using the 7,466 single-cell observations of~\citet{sachs2005}, we infer a
DAG over 11 signaling variables and compare it with the reference network.
Table~\ref{tab:dag-sachs} shows higher F1 and lower SHD and Brier score
for \method{} relative to VTI. In the expanded comparison, \method{} has
the highest F1 and lowest SHD, VTI has the highest AUROC, and raw DiBS has
the lowest Brier. Three of the ten DiBS particles contain directed cycles;
their probability mass is retained and disclosed rather than removed.
Within each \method{} fit, all 2,048 graph draws coincide. Empirical edge
probabilities are binary, and the summed Brier counts the same discrepancies
as SHD. These scores evaluate graph recovery against the reference network.
Across starts, VTI also assigns 97.07--98.88\% of its draws to a single
graph in each fit (Table~\ref{tab:dag-sachs-all}). On the same target,
evidence bounds give graph TV above $0.995$ for one \method{} start and
$0.978$ for one VTI start, with simultaneous 95\% confidence.
In the paired integrated fits, entropy backtracking increases early
graph diversity. An exact graph-conditional update improves the joint
objective at fixed hidden marginals. The probability bounds and
refinement analysis in Appendix~\ref{app:sachs-integration} distinguish
graph concentration, reference recovery and posterior accuracy.

\begin{table}[H]
\caption{\textbf{Graph recovery on the Sachs flow-cytometry data.}
All methods use the same observations and reference graph. \method{} and
VTI are selected by the lowest training loss among three starts; added
baselines use fixed seed 3. Models, priors, and budgets differ.
Each \method{} fit returns one adjacency across all 2,048 draws, making its
summed edge Brier equal to SHD. Bold denotes the best raw score.}
\label{tab:dag-sachs}
\centering
\footnotesize
\begin{tabular}{@{}lrrrr@{}}
\toprule
Method & F1 $\uparrow$ & SHD $\downarrow$ & Brier $\downarrow$ & AUROC $\uparrow$ \\
\midrule
SM-VTI & \textbf{.471} & \textbf{18.00} & 18.00 & .683 \\
VTI & .444 & 25.02 & 24.81 & \textbf{.695} \\
DAGMA & .286 & 25.00 & 25.00 & .581 \\
DiBS$^{\dagger}$ & .247 & 29.20 & \textbf{15.72} & .556 \\
DiBS+ & .333 & 28.00 & 28.00 & .612 \\
JSP-GFN & .289 & 41.96 & 39.60 & .579 \\
\bottomrule
\end{tabular}

\par\smallskip
\parbox{\columnwidth}{\scriptsize
$\dagger$: DiBS retains 3/10 cyclic particles (30\% probability mass).
DiBS+ reweights the same particles; the cyclic particles have zero exported
numerical weight, and its weight ESS is one.}
\end{table}

\MoveToAppendix{
\subsection{Synthetic nonlinear DAG setting and results}
\label{app:dag-setting}
We first describe the shared \method{}/VTI target and paired protocol;
the additional methods and their different settings are described below.
The latent model is $m=(P,U)$, with permutation matrix $P$ and strictly
upper-triangular binary edge matrix $U$, giving adjacency $A=P^\top U P$.
The path starts with the $d-1$ decisions of the Lehmer code for $P$.
\method{} visits entries $U_{ij}$ in increasing child rank $j$,
then in increasing parent rank $i<j$ within each child block.
The target stores the strict upper triangle of $U$ in row order.
The parent of each non-root prefix drops its last decision. A path
terminates after all edge bits have been assigned, giving one path for
each $(P,U)$. Multiple encodings can represent the same DAG, so the
uniform prior on $(P,U)$ induces DAG probabilities proportional to their
numbers of encodings. The continuous parameters have standard-normal priors. The
Gaussian structural equation model uses observation-noise SD one and
width-10 ReLU conditional-mean networks with output intercepts and no hidden
biases; the first node in the ordering has mean zero.

For each of three data seeds, we generate a ten-node DAG with edge
probability $0.5$ and 1,024 observations; smaller $N$ use nested prefixes.
Generating network parameters have independent random signs and magnitudes
uniform on $[0.3,0.7]$. Both methods use the same observations, latent target,
and graph-scoring convention.

\method{} introduces each child block's likelihood-active coordinates after
selecting its parent set. For node $j$ in the sampled order with $k_j$
parents, the synthetic child block introduces $10k_j$ hidden weights, the
output intercept, and ten output weights if $k_j>0$. With fixed context
$h_e$ and already active coordinates $z$, the new Gaussian coordinates are
$a=\mu_e(z,h_e)+\exp(s_e(z,h_e))\odot u$. The lift places the old
coordinates to their typed child slots and places $a$ in the new slots. Its
inverse solves $u=(a-\mu_e)\odot\exp(-s_e)$, with forward log determinant
$\sum_i s_{e,i}$. In Sachs, the child block includes the hidden bias and the
output weights also when the parent set is empty. The model dimension and
root convention are determined by the likelihood.

The original schedule jointly trains five affine
layers of width 128 and the policy using Adam at learning rate
$3\times10^{-4}$. VTI uses the released Affine(5,10) configuration and
optimizers, with a CoSMIC flow mask aligned to the likelihood coordinates.
For $N\leq128$, we also use a stabilized \method{} schedule: the policy stays
uniform for 5,000 flow-only updates, followed by flow Adam at
$3\times10^{-4}$ and policy SGD with base learning rate $10^{-4}$.
Backtracking limits an importance-sampling estimate of the absolute
one-step policy-entropy change, with initial threshold $0.005$.
The SGD rate and entropy threshold decay during the joint phase as specified
in the experiment configuration. The variational family and objective are
unchanged. Figure~\ref{fig:dag-metrics} displays the two schedules as
separate series.

The \method{}/VTI fits in Figure~\ref{fig:dag-metrics} use optimizer
seed 3, float64, and
25,000 updates with 1,024 target evaluations per update. This gives equal
target-evaluation budgets, including the flow-only phase; architectures and
optimizers remain method-specific. The published VTI experiment uses a
50,000-update, multi-start protocol. We evaluate 2,048 sampled adjacency
matrices per fit using the released VTI scoring convention: F1, SHD, and
AUROC are averaged over sampled graphs, while Brier is the summed squared
error of posterior edge probabilities. SHD counts a reversed edge twice;
AUROC is computed over all adjacency entries, including the diagonal.
Figure~\ref{fig:dag-metrics} shows means and sample SDs across three
paired data sets.

\begin{figure}[H]
\centering
\includegraphics[width=\textwidth]{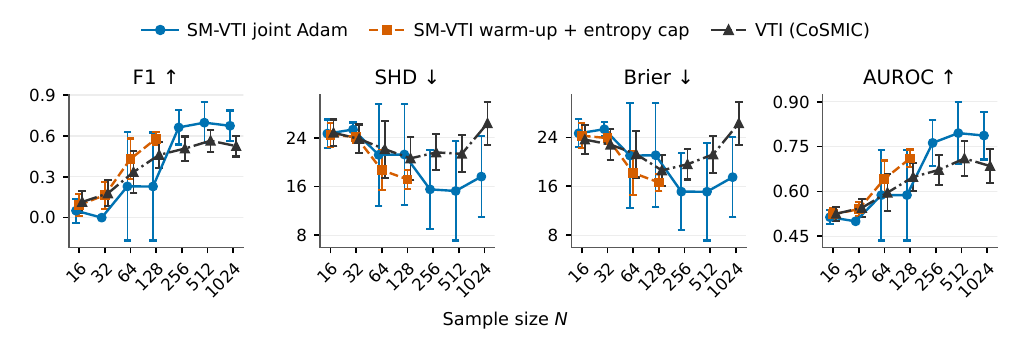}
\caption{\textbf{Synthetic nonlinear DAG recovery under the original training
schedules.} Means and sample standard deviations over three data sets per
sample size. Joint Adam, SM-VTI with flow warm-up and an entropy cap, and VTI
use 25,000 updates and 2,048 graph draws per fit. Warm-up results are
available for $N\leq128$. At $N=16,32$, Table~\ref{tab:dag-small-integrated}
gives the comparisons with exact integration.}
\label{fig:dag-metrics}
\end{figure}

The stabilized schedule has better mean scores on all four metrics than
original joint training at each $N\leq128$, and better scores than VTI at
$N=64,128$. The original schedule yields zero F1 on $2/3$, $3/3$, $2/3$,
and $2/3$ data sets at $N=16,32,64,128$, respectively.
To assess graph concentration, we also cross three optimizer starts
(seeds 3, 17, and 29) with the same three data sets. We count a fit as
concentrated when at least 90\% of its 2,048 graph draws have the same
adjacency. Under the stabilized schedule, the counts are $0/9$ at 10,000
updates for each of these four $N$ values, and $5/9$, $4/9$, $1/9$, and
$0/9$ at 25,000 updates, respectively. Concentration therefore persists at
the smallest sample sizes.

\paragraph{Additional DAG baselines.}
\label{app:dag-baselines}
The added baselines use the same frozen observations and reference graphs
as the displayed \method{}/VTI results, including the nested prefixes and
data seeds 1004, 1000, and 1001. Each added method uses optimizer seed 3
without selecting starts or hyperparameters using graph-recovery scores.
Table~\ref{tab:dag-baseline-protocol} records their settings, and
Table~\ref{tab:dag-additional-baselines} reports all six methods.
At $N=16,32$, both VI methods use exact output integration as in
Table~\ref{tab:dag-small-integrated}. At $N=64,128$, \method{} uses its
stabilized schedule. At $N\geq256$, it uses joint training. The original
training-schedule comparison is in Figure~\ref{fig:dag-metrics}.

\begin{table}[H]
\centering
\small
\caption{\textbf{Settings for the additional nonlinear DAG baselines.}
All use hidden width ten, no hidden or output biases, and fixed seed 3.
DAGMA retains its official sigmoid activation; DiBS and JSP-GFN use ReLU.
The same settings are used for the synthetic and Sachs data.}
\label{tab:dag-baseline-protocol}
\begin{tabular}{@{}p{0.12\textwidth}p{0.49\textwidth}p{0.22\textwidth}p{0.09\textwidth}@{}}
\toprule
\textbf{Method} & \textbf{Training protocol} & \textbf{Graph estimate} & \textbf{Precision} \\
\midrule
DAGMA & Fixed $\lambda_1=0.02$, $\lambda_2=0.005$, threshold $0.3$;
four stages with 50,000-update warm-stage and 80,000-update final-stage
caps, retaining the official convergence checks & One point estimate & float64 \\
DiBS/DiBS+ & 5,000 SVGD steps, 10 particles; 128 likelihood-gradient and
32 acyclicity Monte Carlo draws & Uniform / mixture weights & float64 \\
JSP-GFN & 50,000 updates after 1,000 prefill steps; graph batch 1,024,
one parameter sample, rescaled data minibatch $\min(N,256)$ & 2,048 graphs & float32 \\
\bottomrule
\end{tabular}
\end{table}

DAGMA~\citep{bello2022} returns a point estimate. Its uncalibrated network
weights are not treated as posterior edge probabilities: the thresholded
binary graph is scored once, making its summed Brier equal to its directed
SHD. The regularization and threshold are fixed before graph scoring.
This differs from the truth-informed regularization sweep in the published
VTI baseline comparison and does not reproduce that paper's numerical
DAGMA curve.

DiBS and DiBS+~\citep{lorch2021} use the official joint graph-and-parameter
inference with Gaussian observation variance one and parameter-prior SD
one, RMSprop step size $0.005$, and linear inverse-temperature rates
$\alpha=0.05$ and $\beta=1$. One inference produces the same ten graph/parameter particles for
both methods. DiBS uses uniform weights; DiBS+ uses the official
joint-density mixture weights. F1, directed SHD, and per-graph AUROC are
averaged with these weights. Brier is the summed squared error of the
weighted posterior edge probabilities. No cyclic particle is dropped,
projected, or silently renormalized away.

JSP-GFN~\citep{deleu2023} uses the official nonlinear Gaussian MLP model,
with observation-noise and parameter-prior SD one and Adam learning rate
$10^{-4}$. Its data minibatch
likelihood estimates are rescaled to the full data size. Evaluation draws
the learned graph policy directly, omitting continuous parameter scores
and derivatives that are unnecessary for adjacency-only evaluation;
the graph sampling distribution is unchanged. The float32 choice follows
the upstream JAX default and was fixed from runtime and finite-loss checks
before evaluating graph recovery. VTI/\method{}, DAGMA, and DiBS use
float64; numerical precision is therefore method-specific.

These additional baselines have different targets from the controlled
\method{}/VTI pair. DiBS and JSP-GFN put a uniform prior directly over
DAGs, whereas the pair's uniform ordering-encoding prior weights a DAG by
its number of encodings. The pair retains synthetic output intercepts;
the added no-bias regressions omit both biases. On Sachs, the pair uses
width five, hidden and output biases, and structural penalty 200, while
the added methods use width ten and no such penalty. Training budgets and
Sachs start-selection rules also differ. Accordingly, the expanded
comparison assesses graph recovery under the declared settings and is not
an equal-target or equal-compute inference-family ablation.

\paragraph{Particle diagnostics.}
All 210 synthetic DiBS particles are acyclic, and all new DAGMA and
JSP-GFN graph outputs are DAGs. DiBS+ weights can nevertheless be highly
concentrated: 15 of 22 fits have weight ESS below 1.01, and all six fits
at $N\in\{512,1024\}$ have numerical ESS one. Weight ESS measures the
concentration of the normalized weights, not independence or graph-space
mixing. Point graph-recovery scores do not establish uncertainty calibration.

\begingroup
\small
\setlength{\tabcolsep}{7pt}
\begin{longtable}{@{}rlrrrr@{}}
\caption{Synthetic nonlinear DAG recovery with additional baselines. Entries are mean (sample SD) over three paired data sets per sample size. At $N=16,32$, both VI methods integrate the output layer and use 25,000 updates with shared optimization. At $N=64,128$, SM-VTI uses its stabilized schedule. For $N\geq256$ it uses joint training. Methods share observations and graph-scoring conventions. The additional baselines use their own models and training budgets.}\label{tab:dag-additional-baselines}\\
\toprule
$N$ & Method & F1 $\uparrow$ & SHD $\downarrow$ & Brier $\downarrow$ & AUROC $\uparrow$ \\
\midrule
\endfirsthead
\caption[]{Synthetic nonlinear DAG recovery with additional baselines (continued).}\\
\toprule
$N$ & Method & F1 $\uparrow$ & SHD $\downarrow$ & Brier $\downarrow$ & AUROC $\uparrow$ \\
\midrule
\endhead
\midrule
\multicolumn{6}{r}{Continued on next page}\\
\endfoot
\bottomrule
\endlastfoot
16 & SM-VTI (integrated) & .454 (.051) & 22.889 (1.242) & 14.786 (1.444) & .642 (.026) \\*
 & VTI (integrated) & .439 (.003) & 24.768 (1.424) & 14.245 (.897) & .631 (.005) \\*
 & DAGMA & .474 (.059) & 34.000 (3.606) & 34.000 (3.606) & .643 (.055) \\*
 & DiBS & .155 (.025) & 24.967 (1.320) & 22.130 (1.518) & .533 (.009) \\*
 & DiBS+ & .220 (.037) & 23.463 (1.453) & 22.601 (2.253) & .556 (.012) \\*
 & JSP-GFN & .372 (.111) & 27.556 (3.030) & 21.927 (1.171) & .609 (.044) \\
\addlinespace[3pt]
32 & SM-VTI (integrated) & .537 (.047) & 18.540 (.782) & 15.144 (.831) & .687 (.024) \\*
 & VTI (integrated) & .498 (.017) & 21.181 (.528) & 13.645 (1.384) & .665 (.007) \\*
 & DAGMA & .612 (.063) & 24.000 (3.464) & 24.000 (3.464) & .756 (.044) \\*
 & DiBS & .223 (.048) & 24.333 (1.955) & 20.493 (1.805) & .553 (.017) \\*
 & DiBS+ & .364 (.092) & 20.658 (2.385) & 20.175 (2.196) & .610 (.034) \\*
 & JSP-GFN & .348 (.203) & 32.118 (5.904) & 29.018 (3.757) & .598 (.074) \\
\addlinespace[3pt]
64 & SM-VTI & .431 (.150) & 18.696 (3.279) & 18.199 (3.569) & .642 (.062) \\*
 & VTI & .335 (.152) & 22.111 (4.660) & 21.251 (3.813) & .596 (.063) \\*
 & DAGMA & .659 (.048) & 19.000 (2.646) & 19.000 (2.646) & .781 (.028) \\*
 & DiBS & .314 (.009) & 23.700 (0.794) & 17.750 (0.786) & .582 (.001) \\*
 & DiBS+ & .448 (.041) & 19.656 (1.514) & 19.604 (1.447) & .643 (.018) \\*
 & JSP-GFN & .413 (.086) & 35.276 (4.133) & 29.637 (2.598) & .597 (.059) \\
\addlinespace[3pt]
128 & SM-VTI & .580 (.048) & 17.160 (1.603) & 16.633 (1.443) & .711 (.030) \\*
 & VTI & .461 (.096) & 20.611 (3.555) & 18.605 (2.502) & .648 (.048) \\*
 & DAGMA & .597 (.109) & 19.000 (4.000) & 19.000 (4.000) & .728 (.068) \\*
 & DiBS & .394 (.031) & 22.967 (1.570) & 15.843 (1.167) & .614 (.016) \\*
 & DiBS+ & .490 (.082) & 19.582 (3.393) & 18.844 (2.479) & .662 (.044) \\*
 & JSP-GFN & .227 (.187) & 40.359 (6.052) & 36.979 (4.764) & .500 (.044) \\
\addlinespace[3pt]
256 & SM-VTI & .662 (.125) & 15.530 (6.446) & 15.124 (6.266) & .763 (.078) \\*
 & VTI & .507 (.089) & 21.587 (2.957) & 19.602 (2.496) & .672 (.051) \\*
 & DAGMA & .603 (.096) & 16.667 (4.163) & 16.667 (4.163) & .723 (.051) \\*
 & DiBS & .433 (.014) & 22.733 (0.737) & 13.913 (0.702) & .631 (.007) \\*
 & DiBS+ & .554 (.064) & 18.333 (3.055) & 18.332 (3.054) & .695 (.036) \\*
 & JSP-GFN & .307 (.079) & 44.908 (6.463) & 39.665 (8.481) & .499 (.071) \\
\addlinespace[3pt]
512 & SM-VTI & .697 (.153) & 15.263 (8.119) & 15.106 (7.982) & .795 (.105) \\*
 & VTI & .565 (.081) & 21.354 (3.045) & 21.221 (2.991) & .711 (.058) \\*
 & DAGMA & .636 (.059) & 15.333 (2.887) & 15.333 (2.887) & .741 (.034) \\*
 & DiBS & .448 (.022) & 23.467 (1.710) & 13.320 (0.928) & .637 (.016) \\*
 & DiBS+ & .569 (.067) & 18.333 (3.512) & 18.333 (3.512) & .704 (.040) \\*
 & JSP-GFN & .229 (.239) & 36.720 (8.091) & 33.865 (9.033) & .522 (.111) \\
\addlinespace[3pt]
1024 & SM-VTI & .674 (.112) & 17.621 (6.591) & 17.490 (6.517) & .787 (.080) \\*
 & VTI & .525 (.077) & 26.350 (3.520) & 26.310 (3.503) & .685 (.057) \\*
 & DAGMA & .641 (.024) & 15.333 (1.155) & 15.333 (1.155) & .746 (.018) \\*
 & DiBS & .446 (.007) & 24.467 (0.513) & 13.267 (0.270) & .635 (.002) \\*
 & DiBS+ & .562 (.056) & 19.667 (2.082) & 19.667 (2.082) & .704 (.035) \\*
 & JSP-GFN & .168 (.200) & 31.670 (5.540) & 29.617 (3.844) & .511 (.073) \\
\end{longtable}
\endgroup

}

\MoveToAppendix{
\subsection{Sachs flow-cytometry DAG setting and results}
\label{app:dag-sachs}
We first give the controlled \method{}/VTI comparison, followed by the
additional methods with their different targets and budgets.
The Sachs data~\citep{sachs2005} comprise 7,466 observations of 11 signaling
variables. Following the released VTI preprocessing, we log-transform the
measurements, center each variable, and divide by a common standard deviation
computed over all entries. Both methods use the nonlinear Gaussian DAG model
of \citet{davies2025}, with five hidden units per node regression, hidden and
output biases, and edge-prior penalty $\lambda=200$. VTI uses the released
Affine(5,5) transport with its original CoSMIC flow mask; \method{} uses five
shared affine layers of width 128 on the likelihood-active native coordinates.
The methods share the target and target-evaluation budget, with method-specific
architectures and optimizers.

We fit each method from three optimization starts (seeds 1--3) on the same
data, using 25,000 updates in float64 and 1,024 target evaluations per update.
Within each method, the fit with the lowest training loss at this fixed
endpoint is selected for Table~\ref{tab:dag-sachs}.
Table~\ref{tab:dag-sachs-all} reports all starts. Graph metrics are computed
from 2,048 posterior draws per fit against the literature reference network,
which is used only for evaluation. The published VTI comparison uses ten
starts and 50,000 updates. In our extended runs, the VTI fit selected by
training loss at 50,000 updates failed numerical density evaluation.

\begin{table}[H]
\caption{\textbf{Sachs flow-cytometry results.} Each row reports one of three
25,000-update fits per method on the same data set or their unweighted mean.
Training loss is in units
of $10^3$. Bold seeds identify the fits selected for Table~\ref{tab:dag-sachs}.
Distinct graphs and modal frequency are measured across 2,048 draws.}
\label{tab:dag-sachs-all}
\centering
\scriptsize
\begin{tabular}{@{}llrrrrrrr@{}}
\toprule
\textbf{Method} & \textbf{Seed} & \textbf{Loss} & \textbf{F1} & \textbf{SHD} &
\textbf{Brier} & \textbf{AUROC} & \textbf{Distinct graphs} & \textbf{Mode (\%)} \\
\midrule
\method{} & 1 & 101.899 & .516 & 15.000 & 15.000 & .698 & 1 & 100.00 \\
\method{} & 2 & 101.268 & .500 & 16.000 & 16.000 & .693 & 1 & 100.00 \\
\method{} & \textbf{3} & 101.208 & .471 & 18.000 & 18.000 & .683 & 1 & 100.00 \\
\method{} & Mean & -- & .496 & 16.333 & 16.333 & .691 & -- & -- \\
\addlinespace
VTI & 1 & 103.203 & .390 & 25.024 & 24.801 & .649 & 12 & 98.49 \\
VTI & \textbf{2} & 102.076 & .444 & 25.016 & 24.813 & .695 & 12 & 98.88 \\
VTI & 3 & 103.041 & .248 & 30.069 & 29.645 & .555 & 11 & 97.07 \\
VTI & Mean & -- & .361 & 26.703 & 26.420 & .633 & -- & -- \\
\bottomrule
\end{tabular}
\end{table}
Each \method{} fit has higher F1 and lower SHD than each VTI fit. VTI has
higher AUROC in the training-loss-selected comparison. Within each \method{} fit, all 2,048 sampled
adjacency matrices coincide, so the summed Brier score equals SHD. VTI
also concentrates, with a modal adjacency in 97.07--98.88\% of draws
per fit. Each method has three distinct modal graphs across the starts.
Pairwise SHDs between the modal graphs are 3--5 for \method{} and 14--31
for VTI, using the same edge-disagreement convention. The observed
concentrations and reference-graph scores leave posterior calibration unresolved.

\paragraph{Additional baselines on Sachs.}
The added methods use the same 7,466 preprocessed observations, variable
order, and literature reference graph, with the settings in
Table~\ref{tab:dag-baseline-protocol}. Their fixed seed-3 fits are compared
with the three-start-selected \method{}/VTI endpoints; this is not a
matched multi-start budget. The six-method scores are reported in
Table~\ref{tab:dag-sachs}.

The DiBS output contains three cyclic graphs among its ten particles,
carrying 30\% probability under uniform weights. The raw scores include
these particles, so the DiBS row is not the evaluation of a distribution
supported entirely on DAGs. DiBS+ reweights the same particles: its
exported float64 weights put numerical mass one on a single acyclic
particle and zero on the cyclic particles. Its lower raw cyclic mass
therefore accompanies extreme weight concentration (ESS one), rather than
evidence that the full graph posterior is well explored. These limitations
are disclosed alongside the stronger F1 and SHD of the selected
\method{} fit, VTI's higher AUROC, and DiBS's lower raw Brier.
}

\subsection{Phylogenetic model selection}
\label{sec:phylo-joint}

Following \citet{huelsenbeck2004}, we use Bayesian selection among 203
nucleotide substitution models for the Primates and Vertebrates alignments
on a fixed species phylogeny. Related flow-based variational phylogenetics
models topology and branch lengths~\citep{zhang2020phyloflows}; our discrete
variable selects a substitution-rate partition on the fixed topology.
The models constrain the six exchangeabilities
in different ways, yielding different numbers of free rate parameters
(Appendix~\ref{app:phylo-setting}).

\MoveToAppendix{
\subsection{Phylogenetic setting}
\label{app:phylo-setting}
We use the substitution-model setting of \citet{huelsenbeck2004},
conditioning on a fixed tree topology. The data are the Primates alignment
(12 taxa, 898 mitochondrial-DNA sites) and the Vertebrates
alignment (17 taxa, 432 $\beta$-globin sites). A substitution model $m$
partitions the six nucleotide-exchangeability types
$\mathcal E=\{\mathrm{AC},\mathrm{AG},\mathrm{AT},\mathrm{CG},
\mathrm{CT},\mathrm{GT}\}$ into $K_m$ groups of equal rates. The
$\mathrm{Bell}(6)=203$ models are encoded by restricted-growth strings:
\texttt{111111} denotes the one-rate model and \texttt{123456} the general
time-reversible model.

The continuous parameters are
\[
 \theta_m=(v,\alpha,\pi,\rho_m),\qquad d_m=B+K_m+3,
\]
where $v$ contains the branch lengths, $\alpha$ is the among-site gamma-rate
shape, $\pi$ contains the nucleotide frequencies, and $\rho_m$ contains the
$K_m$ group exchangeabilities with $K_m-1$ free rate parameters. We use
$B=22$ for Primates and $B=31$ for Vertebrates. The Primates representation
retains a degree-two root with separate lengths for its two incident edges;
both methods and the reference chains use this representation.
For $X\ne Y$, the rate matrix has entries $q_{XY}=\rho_{m,g}\pi_Y$ when
$XY$ belongs to group $g$, followed by normalization to unit stationary mean
rate. The likelihood uses Felsenstein pruning \citep{felsenstein1981} with
four discrete-gamma rate categories \citep{yang1994}.

We assign a uniform prior to the 203 models, independent exponential
branch-length priors with rate $10$,
$\pi\sim\mathrm{Dirichlet}(1,1,1,1)$,
$\alpha\sim\mathrm{Uniform}(0,50)$, and the group-size-adjusted Dirichlet
prior for $\rho_m$ of \citet{huelsenbeck2004}. For the fixed topology $s_0$,
the joint posterior is
\[
 p(m,\theta_m\mid x,s_0)\propto
 p(x\mid m,\theta_m,s_0)\,p(\theta_m\mid m)\,p(m),
 \qquad m\in\mathcal M.
\]

The parent rule merges the two highest-labelled groups in the
restricted-growth string. For example, the parent of \texttt{112233}
is \texttt{112222}, whose parent is the one-rate root \texttt{111111}.
The forward path refines one group per step and can stop at any rate
partition. Each dimension-matching map adds one free rate coordinate while retaining the
branch, frequency, and shape coordinates at the lift. For a parent with $K$
groups, let $Y_i=n_i\rho_i/6$ denote the group mass and
$x_i=\log(Y_i/Y_K)$ for $i<K$. The canonical split replaces reference
group $K$ by groups $A$ and $B$ of sizes $n_A,n_B$, using $B$ as the new
reference. The archived transport applies a learned context-dependent
affine transform $a=\mu_e+\sigma_e u$, with $u\sim\mathcal N(0,1)$ and
$\sigma_e>0$, to the split auxiliary. The fixed split then gives
\[
 h=\log(n_A/n_B)+a,\quad
 x'_i=x_i+\log(1+e^h)\ (i<K),\quad x'_K=h.
\]
This divides mass $Y_K$ by the odds $e^h$ and leaves the previous
group masses unchanged in the lift. The inverse subtracts $\log(1+e^h)$
and $\log(n_A/n_B)$ to recover $x$ and $a$, respectively. The fixed split
has zero log absolute determinant, and the auxiliary transform adds
$\log\sigma_e$. A shared conditional state affine transform and coupling
stack act on the child coordinates after each split. The support
transform acts once at the selected terminal model.
}

\MoveToAppendix{
\paragraph{Estimation and evaluation.}
For each alignment, we train ten \method{} fits for 30,000 iterations
with 256 path-RLOO samples per iteration. During the first 2,000 iterations, the model distribution is
uniform and only the conditional transport is updated; subsequent iterations
optimize the joint objective. The transport uses a standard-normal base,
three affine coupling layers of hidden width 64, and a 16-dimensional edge
context. Computation uses float64, with learning rates $10^{-3}$ for the
transport and $5\times10^{-3}$ for the model distribution.

VTI uses an affine CoSMIC flow and a categorical model distribution, with
the same iteration budget, samples per iteration, and uniform-policy warmup.
We select its input score-function learning rate from
$\{10^{-4},3\times10^{-4},10^{-3}\}$ by minimizing the mean training objective
over the final 1,000 iterations of each 10,000-iteration development run,
giving $3\times10^{-4}$ for Primates and $10^{-3}$ for Vertebrates.

For each alignment, the reference combines three independent RJMCMC chains
of one million states each, discarding the first 20\% of each chain. Model TV and top-$k$ agreement use
the pooled model frequencies. Joint and conditional NLL use 192 common
reference states, with 64 selected from each retained chain.

\paragraph{Reference-chain diagnostics.}
The largest TV between one chain's model frequencies and the pooled
frequencies is $0.031$ for Primates and $0.053$ for Vertebrates. The
pooled modal-model indicator has split $\widehat R=1.001$ and $1.002$
and a summed within-chain ESS of 2,210 and 2,732, respectively.

\paragraph{Historical reweighting diagnostic for SM-VTI.}
Using the conditionals from the original \method{} fits, we estimate the 203 model evidences by
importance sampling and normalize them as in
\eqref{eq:importance-model-marginal}, leaving the conditional densities
unchanged. When nonfinite log weights occur, the diagnostic averages use
only finite weights and the model fails the evidence gate. With $S$ finite
weights $w_s$, the conditional ESS fraction is
$(\sum_s w_s)^2/(S\sum_s w_s^2)$. We require a fraction at least $0.01$
and relative standard error at most $0.02$, using 16 block evidence
estimates at each sample size (4,096 to 65,536 draws). Relative standard
error is the standard error of the block estimates divided by their mean.
The mean absolute error of the top-model probability
relative to RJMCMC decreases from $0.0419$ to $0.0021$ for Primates and from
$0.0335$ to $0.0184$ for Vertebrates. In every fit, the top-model evidence
estimate passes the CESS/RSE checks, while some other model evidence
estimates fail. Normalization includes these estimates, so this diagnostic
evaluates the original \method{} fits at limited precision. Appendix~\ref{app:primates-is-pilot} gives paired
evidence estimates for both transports.
}

With the same estimator, optimizer, rates and target-call budget on
Primates, \method{} has lower held-out conditional and joint NLL in all five
paired starts. Mean conditional NLL is $-71.751$ versus $-71.518$ and
mean joint NLL is $-69.117$ versus $-68.879$ for \method{} and CoSMIC,
respectively. CoSMIC has lower TV for the learned model probabilities in
four of the five pairs (means 0.0596 versus 0.0829).
Appendix~\ref{app:primates-matched} gives the full matched comparison and
reports post-fit evidence estimation as a separate diagnostic.

Under the original method-specific recipes, \method{} has lower mean
model TV on both alignments and identifies the RJMCMC top model in every
fit (Table~\ref{tab:phylo-main}). The original Primates TV gap of 0.0686
versus 0.8127 is sensitive to training.

\begin{table}[!htbp]
\caption{\textbf{Phylogenetic model selection under the original training
recipes.} The mean (SD) of model TV and
joint NLL obtained with \method{} and VTI on the Primates (P) and
Vertebrates (V) datasets, over ten fits unless marked with $\dagger$.
\lq Top-1\rq{} counts the agreements with the RJMCMC top models.
$\dagger$: NLL is reported for nine fits (numerical density evaluation
failed for one fit).}
\label{tab:phylo-main}
\centering
\scriptsize
\begin{tabular}{@{}llrrr@{}}
\toprule
\textbf{Data} & \textbf{Method} & \textbf{TV} & \textbf{Joint NLL} & \textbf{Top-1} \\
\midrule
P & \method{} & \best{.0686\,(.0105)} & $-69.069$ (.218) & \textbf{10/10} \\
  & VTI & .8127 (.1095) & $-59.100$ (3.129)$^\dagger$ & 4/10 \\
\addlinespace
V & \method{} & \best{.0772\,(.0061)} & \best{-97.825\,(.225)} & \textbf{10/10} \\
  & VTI & .4384 (.0824) & $-97.328$ (.288) & \textbf{10/10} \\
\bottomrule
\end{tabular}
\end{table}

\section{Discussion}

\method{} uses a construction tree to define a joint variational approximation
on model-specific parameter spaces, using structural relationships between
models. The experiments show competitive or improved posterior
approximation relative to VTI, with trade-offs across data sizes and graph
metrics. The added application-specific DAG methods are stronger on some
metrics and sample sizes; graph recovery alone does not establish posterior
uncertainty calibration.

The current formulation requires a finite rooted construction with one path
per model and a tractable invertible dimension-matching map on each edge.
Representing multiple valid constructions per model calls for path mixtures.
Joint reverse-KL training can also favor models with smaller conditional
approximation gaps \eqref{eq:conditional-gap-tilt}. A second direction is
therefore to explore alternative training objectives that reduce the
resulting bias in posterior model probabilities.

\MoveToAppendix{
\begin{figure}[!htbp]
\centering
\includegraphics[width=.92\textwidth]{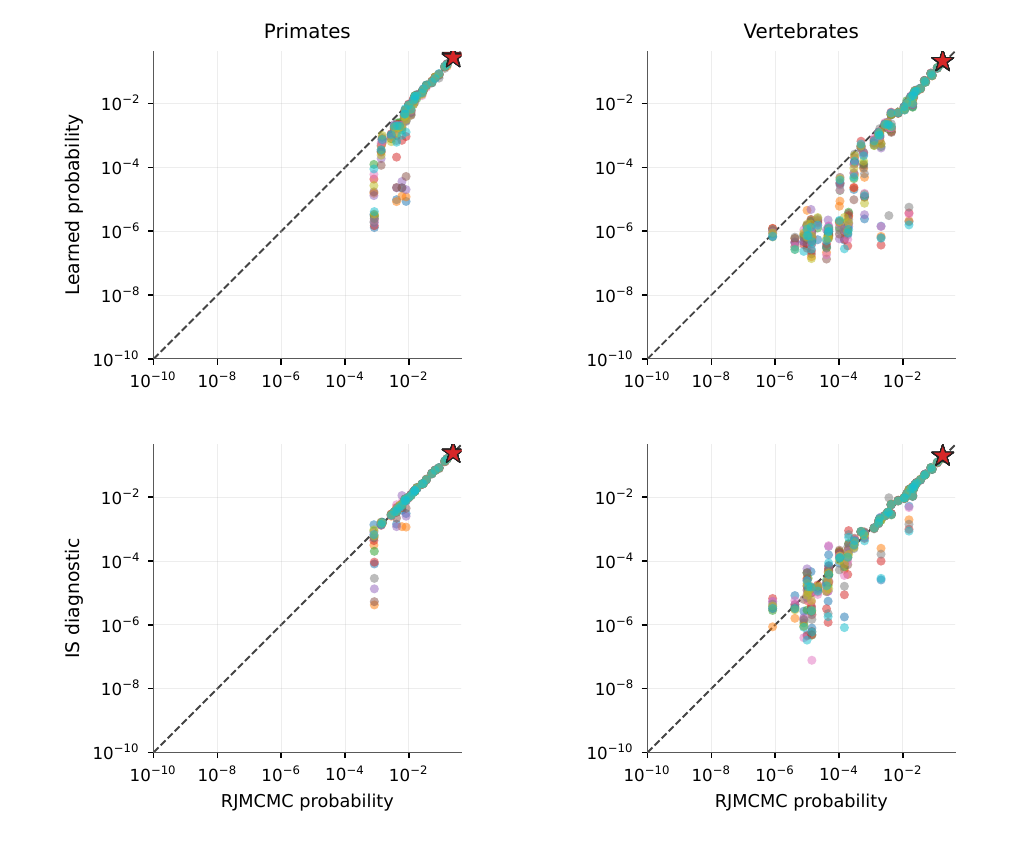}
\caption{\textbf{Importance-sampling diagnostic for the original \method{}
fits.} Columns show Primates and Vertebrates. Rows show the fitted model
probabilities $q_\psi(m)$ and the normalized importance-sampling estimates
$\widehat\pi_M^{\mathrm{IS}}(m)$. Colours identify fits, stars mark the RJMCMC
top model, and the dashed line denotes agreement with the RJMCMC frequencies.
In the lower row, the complete set of model-evidence estimates fails the
CESS/RSE checks in every fit. Appendix~\ref{app:primates-is-pilot} gives
the paired Primates analysis for both transports.}
\label{fig:phylo-primates}
\end{figure}
}

\MoveToAppendix{
\begin{figure}[!htbp]
\centering
  \includegraphics[width=.92\textwidth]{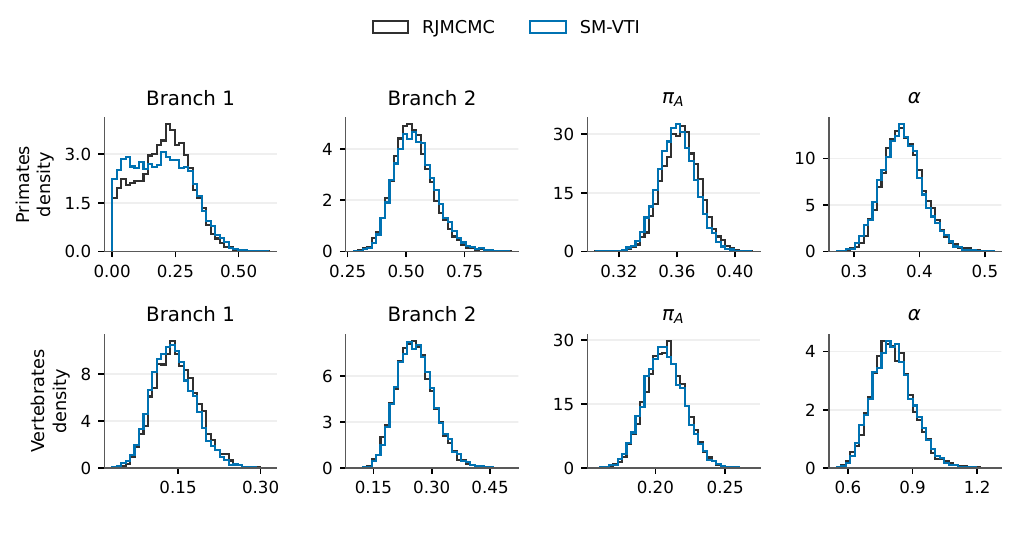}
\caption{\textbf{Illustrative conditional marginals under the RJMCMC top
model at 30,000 iterations.} Top: Primates, model \texttt{121121}, seed 1810.
Bottom: Vertebrates, model \texttt{123343}, seed 1217. Histograms compare
8,192 draws per distribution for two branch lengths, a nucleotide
frequency, and the gamma-rate shape. Reference draws are sampled from
the pooled retained RJMCMC states conditional on the labeled model.}
\label{fig:phylo-conditionals}
\end{figure}
}

\FloatBarrier
\section*{AI Use Statement}
Generative AI assisted with coding, language polishing, and mathematical checks.
The authors reviewed the AI-assisted work and take full responsibility for the
manuscript.

\bibliography{references}

\section*{Reproducibility Checklist}

The checklist records the experimental materials described in this study. The accompanying manuscript-source archive contains the LaTeX source, bibliography, tables, and figures only; experimental code and companion research artifacts are not included in that archive.

\begin{enumerate}
  \item For all models and algorithms presented, check if you include:
  \begin{enumerate}
    \item A clear description of the mathematical setting, assumptions, algorithm, and/or model. [Yes; Sections 2--3 and Figure 1.]
    \item An analysis of the properties and complexity (time, space, sample size) of any algorithm. [Yes; Section 3.4 and Appendix B.3 give time, target-evaluation, and working-memory bounds for path sampling and exact enumeration.]
    \item (Optional) Source code, with specification of all dependencies, including external libraries. [Yes; the accompanying experiment source archive contains experiment-specific entry points, configurations, dependency specifications, and a checksum-verified VTI source export. DAG and VS/phylogenetic environments are documented separately.]
  \end{enumerate}

  \item For any theoretical claim, check if you include:
  \begin{enumerate}
    \item Statements of the full set of assumptions of all theoretical results. [Yes; Sections 3 and Appendices A--B state the tree and parameterization assumptions, change-of-variables conditions, support and moment conditions, and conditions for differentiating expectations.]
    \item Complete proofs of all theoretical results. [Yes; Supplementary Sections A, B.2, and B.4.]
    \item Clear explanations of any assumptions. [Yes; Appendices A.2 and B.2--B.4 explain the roles of the support, regularity, and integrability conditions in the density, gradient, and model-mass identities.]
  \end{enumerate}

  \item For all figures and tables that present empirical results, check if you include:
  \begin{enumerate}
    \item The code, data, and instructions needed to reproduce the main experimental results (either in the supplemental material or as a URL). [Yes; the source archive provides training and evaluation entry points, frozen inputs, compact per-fit results, and reconstruction instructions. Large historical reference and checkpoint files are supplied in separate companion archives. The reproduction guide distinguishes reconstruction from saved outputs, checkpoint evaluation, and new training.]
    \item All the training details (e.g., data splits, hyperparameters, how they were chosen). [Yes; the experiment-specific configurations record architectures, optimization schedules, budgets, seeds, numerical precision, and selection rules. The phylogenetic entry points retain the historical 22-branch Primates target used in the paper.]
    \item A clear definition of the specific measure or statistics and error bars (e.g., with respect to the random seed after running experiments multiple times). [Yes; Section 4 defines the measures, and table captions state the summary statistics.]
    \item A description of the computing infrastructure used. (e.g., type of GPUs, internal cluster, or cloud provider). [Yes; experiments use an Ubuntu 24.04.2 LTS server with two AMD EPYC 9754 128-core processors, 1~TiB RAM, and three NVIDIA L40S GPUs. Appendix G describes timing and CPU-thread settings. The original VTI/SM-VTI GPU fits use float64; their DAG environment uses PyTorch 2.4.1 with CUDA 12.1, while the VS and phylogenetic environment uses PyTorch 2.9. Added DAGMA fits use PyTorch 2.9.0/CUDA 12.9, and DiBS/JSP-GFN use JAX/JAXlib 0.5.3. Added JSP-GFN fits use float32. Appendix C records the additional baseline settings and diagnostics.]
  \end{enumerate}

  \item If you are using existing assets (e.g., code, data, models) or curating/releasing new assets, check if you include:
  \begin{enumerate}
    \item Citations of the creator If your work uses existing assets. [Yes; Section 4 cites the benchmark and model sources, and the source archive retains third-party attribution and provenance.]
    \item The license information of the assets, if applicable. [No; code licenses and benchmark sources are documented. Sachs observations have a CC BY 4.0 source; the Primates alignment matches the GPL-3.0-or-later MrBayes example. Both fixed tree topologies are author-specified experimental configurations. Terms for the Vertebrates alignment and bundled tree file, and the CDT Sachs reference graph remain unconfirmed.]
    \item New assets either in the supplemental material or as a URL, if applicable. [Yes; the supplementary appendices include additional figures and numerical results.]
    \item Information about consent from data providers/curators. [Not Applicable; the benchmark data are public scientific data.]
    \item Discussion of sensible content if applicable, e.g., personally identifiable information or offensive content. [Not Applicable.]
  \end{enumerate}

  \item If you used crowdsourcing or conducted research with human subjects, check if you include:
  \begin{enumerate}
    \item The full text of instructions given to participants and screenshots. [Not Applicable.]
    \item Descriptions of potential participant risks, with links to Institutional Review Board (IRB) approvals if applicable. [Not Applicable.]
    \item The estimated hourly wage paid to participants and the total amount spent on participant compensation. [Not Applicable.]
  \end{enumerate}
\end{enumerate}

\clearpage
\appendix
\onecolumn
\raggedbottom

\section{Tree factorization and path-RLOO}

\subsection{Representation of finite model distributions}
\label{app:tree-policy-representation}

The tree factorization in \eqref{eq:tree-model-distribution} is expressive
when its local categorical probabilities are unconstrained.

\begin{proposition}
\label{prop:tree-policy-representation}
Suppose every construction state has at least one designated model state in
its descendant subtree.  For every strictly positive mass function $\mu$ on
$\Mcal$, there are local stopping and child-selection probabilities for which
$q_\psi(m)=\mu(m)$ for all $m\in\Mcal$.  The local probability simplexes have
$K-1$ degrees of freedom in total.
\end{proposition}

\begin{proof}
Write $v\preceq w$ when $v$ lies on the path from the root to $w$, and let
$M_v=\sum_{m:\,v\preceq v_m}\mu(m)$.  Then $M_r=1$ and $M_v>0$.  For a
child $c$ of $v$, set
\[
 q_\mu(c\mid v)=\frac{M_c}{M_v},
 \qquad
 q_\mu(\texttt{stop}\mid v_m)=\frac{\mu(m)}{M_{v_m}}.
\]
The child subtrees, together with $v_m$ itself when $v$ is a model state,
partition the model states below $v$, so these probabilities normalize.  Along
$G_m$, the factors telescope:
\[
 q_\mu(m)=\frac{\mu(m)}{M_{v_m}}
           \prod_{(v,c)\in G_m}\frac{M_c}{M_v}
         =\mu(m).
\]
Let $\mathrm{ch}(v)$ denote the children of $v$.  The local simplex has
dimension $|\mathrm{ch}(v)|$ at a model state and
$|\mathrm{ch}(v)|-1$ otherwise.  Since
$\sum_v|\mathrm{ch}(v)|=|\mathcal V|-1$, their total dimension is
$(|\mathcal V|-1)-|\mathcal V\setminus\mathcal V_{\mathrm M}|=K-1$.
Strict positivity makes the construction realizable by finite logits; boundary
mass functions follow by limits.
\end{proof}

Thus an explicit finite tree is a factorization of a categorical distribution,
not a reduction in its dimension.  The shared logit function in
\eqref{eq:implicit-sequential-policy} trades this exact representation property
at fixed capacity for parameter sharing across partial inclusion vectors.

\subsection{Unbiased path-RLOO gradient}
\label{app:tree-policy-gradient}

For fixed transport parameters $\phi$, define
\begin{align*}
 c_\phi(m)
 &=\E_{q_\phi(\cdot\mid m)}
 \left[\log q_\phi(\theta_m\mid m)-\log\eta_m(\theta_m)\right],\\
 h_{\psi,\phi}(m)
 &=\log q_\psi(m)-\log p(m)+c_\phi(m).
\end{align*}
The $\psi$-dependent part of \eqref{eq:joint-loss} is
$\mathcal L_{\mathrm{joint}}(\psi,\phi)=
\sum_{m\in\Mcal}q_\psi(m)h_{\psi,\phi}(m)$.  Hence
\begin{align}
 \nabla_\psi\mathcal L_{\mathrm{joint}}
 &=\E_{m\sim q_\psi}
 \left[\{h_{\psi,\phi}(m)+1\}
 \nabla_\psi\log q_\psi(m)\right].
 \label{eq:tree-policy-score-identity}
\end{align}
Because $\E_{m\sim q_\psi}\nabla_\psi\log q_\psi(m)=0$, the unit term has
zero expectation, and therefore
\begin{equation}
 \nabla_\psi\mathcal L_{\mathrm{joint}}
 =\E_{m\sim q_\psi}
 \left[h_{\psi,\phi}(m)
 \nabla_\psi\log q_\psi(m)\right].
 \label{eq:tree-policy-shifted-score-identity}
\end{equation}

\begin{proposition}[Unbiased path-RLOO score gradient]
\label{prop:path-rloo}
Suppose $q_\phi(\theta_m\mid m)$ has no $\psi$-dependent parameters,
$q_\psi(m)>0$ is differentiable for every model, and
\[
 \E\!\left[
 |\ell(m,\theta_m)|
 \left\{1+\left\|\nabla_\psi\log q_\psi(m)\right\|\right\}
 \right]<\infty.
\]
For $B\geq2$ independent path--parameter draws, the estimator in
\eqref{eq:path-decision-gradient}, with $\operatorname{sg}$ treated as a
stop-gradient, is unbiased for $\nabla_\psi\mathcal L_{\mathrm{joint}}$.
\end{proposition}

Absolute integrability ensures that the conditional means and leave-one-out
baselines exist. Score-weighted integrability makes the gradient estimator
integrable; independence makes each baseline have zero expected product with
the held-out score.

\begin{proof}
For a sampled model $m_b$, the reparameterized conditional draw in
\eqref{eq:path-sampled-loss-term} has conditional expectation
$h_{\psi,\phi}(m_b)$.  Thus the unbaselined score term has expectation
\eqref{eq:tree-policy-shifted-score-identity}.  For $B\geq2$, define
\[
 \bar\ell_{-b}=\frac{1}{B-1}\sum_{j\ne b}\ell_j.
\]
Conditional on all draws other than $b$, the detached value $\bar\ell_{-b}$ is
independent of $m_b$.  Therefore
\begin{align*}
 &\E\Bigl[\bar\ell_{-b}\nabla_\psi\log q_\psi(m_b)
 \,\Bigm|\,\{m_j,\theta_j\}_{j\ne b}\Bigr]\\
 &=\bar\ell_{-b}\sum_{m\in\Mcal}q_\psi(m)
   \nabla_\psi\log q_\psi(m)\\
 &=\bar\ell_{-b}\nabla_\psi\sum_{m\in\Mcal}q_\psi(m)=0.
\end{align*}
Subtracting the leave-one-out baseline changes the variance but not the
expectation.  Averaging over $b$ proves the claim.
\end{proof}

\section{Supplementary method details}
\label{app:relocated-material}
\the\supplementarymaterial

\FloatBarrier
\section{Additional phylogenetic results}
\label{app:phylo-results}

The original 30,000-update fits use 256 target calls per update
and the target and RJMCMC reference described in Appendix~\ref{app:phylo-setting}.
Table~\ref{tab:phylo-main} gives their model TV and joint NLL.
Mean conditional NLL is $-71.710$ versus $-66.638$ on Primates
and $-100.505$ versus $-100.791$ on Vertebrates for \method{} and VTI.
The Primates VTI NLL summarizes the nine of ten fits that pass density
evaluation. Appendix~\ref{app:primates-matched} reports the separate
common-optimizer experiment.

\section{Model-distribution factorization and edge context}
\label{app:structure-ablation}

We examine model-distribution factorization and edge context on the 128-model
robust variable-selection problem. The comparison uses either the local tree
probabilities in \eqref{eq:tree-model-distribution} or separate categorical
logits for the terminal models, with or without edge context. All variants
share the target, native coordinates, map and transport architectures, exact
joint objective, and target-evaluation budget. The no-context variants set
the edge-context inputs to the learned birth and shared coupling networks
to zero, preserving their input dimensions and parameter counts.

Table~\ref{tab:structure-ablation} reports paired differences from \method{}.
With context, the two factorizations have similar mean model TV, and median
paired NLL differences are near zero. Both no-context variants increase mean
TV and yield positive mean and median NLL differences under Medium and High
misspecification. The large gaps between mean and median NLL differences,
particularly for the tree without context, reflect the influence of extreme
fits. Both factorizations benefit from edge context in these comparisons.

\begin{table}[H]
\centering
\caption{\textbf{Effects of model-distribution factorization and edge context.}
Paired differences at 30,000 iterations over ten data sets: mean [median]
for joint NLL and mean for model TV. $\Delta$ is the alternative minus
\method{}, so positive values favor \method{}.}
\label{tab:structure-ablation}
\small
\begin{tabular}{llrr}
\toprule
\textbf{Setting} & \textbf{Alternative} & \textbf{\boldmath $\Delta$ Joint NLL} & \textbf{\boldmath $\Delta$ TV} \\
\midrule
Medium & Categorical model distribution + context & .023 [$-$.001] & $-$.0010 \\
& Tree distribution, no context & 22.701 [1.092] & .0576 \\
& Categorical model distribution, no context & 2.113 [1.221] & .0557 \\
\addlinespace
High & Categorical model distribution + context & $-$.894 [$-$.037] & $-$.0047 \\
& Tree distribution, no context & 1922.688 [3.264] & .0860 \\
& Categorical model distribution, no context & 11.233 [3.849] & .0769 \\
\bottomrule
\end{tabular}
\end{table}

\subsection{Sensitivity to construction order}
\label{app:order-sensitivity}

On three paired Medium data sets we compare ascending, reverse, and
prespecified random predictor orders $(3,5,7,4,6,1,2)$. Canonical model
identities and physical coefficient positions stay fixed. After inserting
a new coefficient, a fixed permutation puts it in the child's canonical
slot. The permutation has unit absolute determinant. The order changes
the prefix composition and the local branching probabilities together.

Each order starts from the same terminal masses, chosen from the ascending
tree with zero local logits. The experiment pairs data seeds 1000--1002
and optimizer seeds 6101--6103. Fits share 15,795 learned parameters,
float64 arithmetic, Adam at $10^{-3}$ with gradient clipping at 20, and
30,000 updates of eight draws per model. Evaluation uses the same
held-out RJMCMC reference within each data set.

Mean model TV is $0.2805$, $0.2703$, and $0.2711$ in ascending, reverse,
and random order. Mean joint NLL is $5.4858$, $5.4688$, and $5.4777$,
respectively. Reverse order has lower TV in all three pairs, and the
joint-NLL direction varies across data sets (Figure~\ref{fig:order-sensitivity}).
This measures sensitivity within a monotone-addition family. The branching
comparison in Appendix~\ref{app:topology-sensitivity} changes which
models share a prefix by changing the parent rule.

\begin{figure}[H]
\centering
\includegraphics[width=\textwidth]{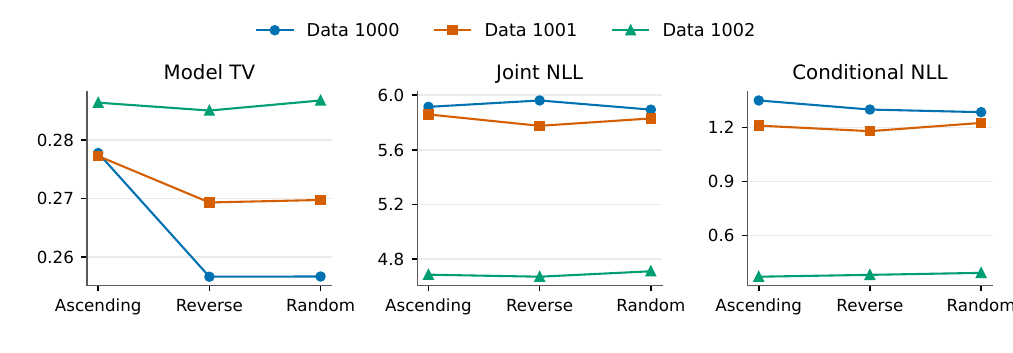}
\caption{\textbf{Paired construction-order comparisons.} Each line follows
one frozen Medium data set with its paired optimizer seed across ascending,
reverse, and fixed random orders. The panels report model TV, held-out joint
NLL, and conditional NLL at the fixed endpoint; lower values are better.
Observed training times are retained in the accompanying data; they were
measured on a shared server and do not establish an isolated speed ranking.
The three orders use the same target, canonical model identifiers, initial
terminal masses, architecture, and evaluation budget; their policy
factorization and transport composition both change.}
\label{fig:order-sensitivity}
\end{figure}

\section{Reproducibility}
\label{app:reproducibility}

The accompanying reproduction material records the input data, source
revisions, run specifications, and SHA-256 digests for the reported results and
reference chains.

\FloatBarrier
\section{Controlled comparisons and additional diagnostics}
\label{app:additional-controls}

\subsection{Conditioner context, learned maps, and construction order}
\label{app:additional-vs}
\paragraph{Native and saturated conditional representations.}
\label{app:saturated-equivalence}
Let $D=\max_{m\in\Mcal}d_m$. For a fixed model $m$, write $\xi_m$ for the
root and edge noises in \eqref{eq:path-map}, with density $r_m$ in the
native base coordinates. Their total dimension is
$d_r+\sum_{i=1}^{L_m}(d_{c_i}-d_{v_i})=d_m$, where
$e_i=(v_i,c_i)$. In the saturated space we group the coordinates into
$z\in\R^{d_m}$ and $a\in\R^{D-d_m}$ by the model-specific permutation.
The common-reference assumption means that $z$ has density $r_m$ and
$a$ independently has density $\nu^{\otimes(D-d_m)}$. The two constructions
use the same support chart $\chi_m$ for the active coordinates.

\begin{proposition}[Mutual representation with unrestricted maps]
\label{prop:unrestricted-native-saturated}
Assume the common references and charts above, with the regularity
conditions of Appendix~\ref{app:path-regularity}. Consider saturated
transports with an arbitrary model-indexed active-block bijection and
identity on the inactive block, as in the factorization
\eqref{eq:cosmic-factorization}. Consider native transports on a fixed
admissible construction tree with an arbitrary model-indexed terminal
bijection $F_m^{\mathrm{term}}$. These two classes represent the same
families of model-conditional densities on the native supports.
\end{proposition}

\begin{proof}
For each model, let $C_m:\xi_m\mapsto z_{L_m}$ denote the root and edge
composition in \eqref{eq:path-map} up to the terminal flow. Each edge
admits an inverse that recovers both the previous state and the noise
introduced there, so reversing the path makes $C_m$ a bijection between
native base and unconstrained terminal coordinates.

Take a native transport and let $K_m=F_m^{\mathrm{term}}\circ C_m$.
In grouped saturated coordinates use
\[
 {\widetilde K}_m(z,a)=(K_m(z),a).
\]
The inverse is $(K_m^{-1},\mathrm{Id})$ acting blockwise, and the absolute
Jacobian determinant is $|\det DK_m(z)|$. Grouping and restoring the
coordinate order by permutations has unit absolute determinant.
Applying $\chi_m$ to the active output and using change of variables
gives the saturated density
\[
 \widetilde q(\theta_m,a\mid m)
   =q_\phi(\theta_m\mid m)\,\nu^{\otimes(D-d_m)}(a).
\]
Integrating over $a$ recovers the native conditional exactly.

Conversely, let $S_m$ be the active-block bijection of a saturated
transport. Fix an admissible choice of root and edge maps on the entire
construction tree. For each model set
\[
 F_m^{\mathrm{term}}=S_m\circ C_m^{-1}.
\]
Then $\chi_m\circ F_m^{\mathrm{term}}\circ C_m=\chi_m\circ S_m$ pointwise.
With the same base density $r_m$, the two constructions give the same
native law after marginalizing the inactive reference variables.
\end{proof}

The equivalence also applies to the joint density when the two
constructions use the same model masses. Proposition~\ref{prop:tree-policy-representation}
gives every strictly positive mass vector with free branch and stop
probabilities, as for categorical logits. For $|\mathcal V|$ states and
$K$ allowed stops, the sum of the local simplex dimensions is
$(|\mathcal V|-1)+K-|\mathcal V|=K-1$. Finite network parameterizations
restrict the maps they can implement. In particular, the second construction
in the proof uses the freedom to choose the terminal map separately for
each model.

\paragraph{Sharing network weights and sharing path prefixes.}
\label{app:prefix-sharing}
The shared weights $\omega$ are used to compute each edge flow
$F_{\omega,e}(\cdot;h_e)$. Local maps can differ across edges because
their active coordinates and contexts differ. The full-path CoSMIC
control also shares conditioner weights across models. Its conditioner
takes the construction path as features of the model-dependent saturated
transport.

In \method{}, the same edge $e=(v,c)$ uses the same lift, flow weights,
and edge context whenever it is traversed. Thus it uses the same map
\[
 T_e(z_v,u_e)=F_{\omega,e}\bigl(\mathcal B_e(z_v,u_e);h_e\bigr).
\]
If paths share a prefix $\mathcal P$ ending at state $v$, let
$\xi_{\mathcal P}$ collect the root and prefix noises. Write
$G_{\mathcal P}(\xi_{\mathcal P})=z_v$ for the common root and edge
composition through $v$. For each model $m$ reached along $\mathcal P$,
write $\zeta_m$ for the remaining noises and $H_{m\mid\mathcal P}$ for
the remaining edge maps, terminal flow, and support chart. Then
\[
 \Phi_{\phi,m}(\xi_{\mathcal P},\zeta_m)
   =H_{m\mid\mathcal P}\bigl(G_{\mathcal P}(\xi_{\mathcal P}),\zeta_m\bigr).
\]
For instance, paths to $\{1,2\}$ and $\{1,3\}$ both pass through
$\{1\}$ and compose the same map up to that state. If the draws are
coupled to have the same $\xi_{\mathcal P}$, their intermediate states
at $v$ coincide. The subsequent maps can change those coordinates in
a model-specific way.

\paragraph{How the shared prefix enters the variational density.}
Let $r_{\mathcal P}$ be the reference density of $\xi_{\mathcal P}$
and $r_{m\mid\mathcal P}$ the product density of the remaining noises.
The latter is one at a stop with no further noises. The same prefix
has the same intermediate density in each descendant transport,
\[
 q_{\phi,\mathcal P}(z_v)
   =r_{\mathcal P}\bigl(G_{\mathcal P}^{-1}(z_v)\bigr)
       \left|\det D G_{\mathcal P}^{-1}(z_v)\right|.
\]
The prefix and remaining noises are independent conditional on the
model. Hence, writing $(z_v,\zeta_m)=H_{m\mid\mathcal P}^{-1}(\theta_m)$,
change of variables gives
\[
 q_\phi(\theta_m\mid m)
 =q_{\phi,\mathcal P}(z_v)\,r_{m\mid\mathcal P}(\zeta_m)
       \left|\det D H_{m\mid\mathcal P}^{-1}(\theta_m)\right|.
\]
Thus the shared intermediate law enters each model's conditional density
through its own suffix inverse. This is the structural inductive bias in
the finite transport: fitting the prefix changes a common intermediate
law used by multiple models. The joint approximation uses these
conditionals with $q_\psi(m)$ in the same reverse-KL objective
\eqref{eq:joint-loss}.

Table~\ref{tab:vs-path-matched-capacity} compares finite architectures
with the same path information, common training, and similar parameter
counts. The ablations below change the local map or the parent rule
within \method{}. Table~\ref{tab:structure-ablation} separately compares
the discrete parameterizations and the use of edge context.

\paragraph{Learned maps and direct append.}
The direct-append control uses $\mathcal B_e(z,u)=(z,u)$.
On ten Medium data sets it shares the tree, edge context, initialization
of the common weights and terminal masses, and 30,000 updates of 1,024
target evaluations with the learned-map fit. We use both the same
downstream width of 48 and a width of 63 to match the total parameter
count. The learned map gives lower model TV and joint NLL in all ten
pairs against each control. Both extreme conditional NLLs in the
width-48 append arm are retained. Under High misspecification, four of
ten wider-append fits have extreme conditional NLL at the endpoint.
The Medium result measures the benefit of learning the lift for this
particular path, context, and flow architecture.

\begin{table}[H]
\centering
\small
\caption{\textbf{Local-map ablation on ten paired Medium data sets.}
Scores are means after 30,000 updates. Widening the append control
approximately matches the learned-map parameter count.}
\label{tab:vs128-lift-append}
\begin{tabular}{@{}lrrr@{}}
\toprule
Local map & Parameters & Model TV & Joint NLL \\
\midrule
Learned lift & 15,795 & 0.252 & 5.665 \\
Direct append, same width & 10,595 & 0.343 & 6.572 \\
Direct append, comparable total size & 15,725 & 0.339 & 6.506 \\
\bottomrule
\end{tabular}
\end{table}

\paragraph{Branching sensitivity at fixed model space.}
\label{app:topology-sensitivity}
In addition to changing predictor orders (Appendix~\ref{app:order-sensitivity}),
we compare two parent rules on ten paired data sets per setting. The
ascending rule removes the largest predictor index. The alternate rule
removes the smallest index from even-sized subsets excluding predictor 7
and removes the largest index otherwise. Both retain model depths and
coefficient positions. The two rules put 64 and 40 models, respectively,
in their largest root subtrees.

Within each pair the target, reference, initial terminal model
probabilities, 15,795 learned parameters, Adam settings, and 30,000
updates of 1,024 target calls agree. The median paired joint NLL changes
are small and have opposite signs on Medium and High
(Table~\ref{tab:topology-sensitivity}). One High pair changes by $-4.40$
and contributes $-0.44$ to the ten-pair mean. Some High pairs have
larger TV changes. Each rule induces a different transport on the same Bayesian model.

\begin{table}[H]
\centering
\small
\caption{\textbf{Sensitivity to changing the parent rule.}
Each setting uses ten paired data sets. Differences are alternate minus
ascending. TV refers to each fit's error relative to the same reference.}
\label{tab:topology-sensitivity}
\begin{tabular}{@{}lrrr@{}}
\toprule
Data & Median change in joint NLL & Median absolute TV change & Max.\ absolute TV change \\
\midrule
Medium & $-0.0024$ & 0.0050 & 0.0123 \\
High & $+0.0338$ & 0.0043 & 0.0420 \\
\bottomrule
\end{tabular}
\end{table}

\subsection{Primates with matched optimization}
\label{app:primates-matched}
We refit the two transports on the same 203-model target using categorical
logits, the same RLOO score gradient, Adam with flow and policy learning
rates $10^{-3}$ and $5\times 10^{-3}$, a 2,000-update flow-only warm-up,
and 30,000 updates of 256 draws each. We selected the common policy rate
by training objective on a separate development fit. Across five paired
optimizer starts (1821--1825), ten fixed endpoints pass density checks.
Training and evaluation use the same objective and held-out reference for
both methods. Transport parameter counts are 92,570 and 123,772.

\begin{table}[t]
\centering
\small
\caption{\textbf{Matched training on Primates.} Means over five paired
optimizer starts on one alignment. Negative log densities use common
reference draws, and TV uses the learned model probabilities.
Time includes training setup and warm-up. CUDA memory peaks are in MiB.}
\label{tab:primates-additional}
\begin{tabular}{@{}lrr@{}}
\toprule
 & \method{} & CoSMIC \\
\midrule
Conditional NLL & $-71.751$ & $-71.518$ \\
Joint NLL & $-69.117$ & $-68.879$ \\
Learned model TV & 0.0829 & 0.0596 \\
Active training minutes & 59.02 & 43.10 \\
Peak allocated MiB & 996.45 & 986.10 \\
Peak reserved MiB & 1220 & 1230 \\
\bottomrule
\end{tabular}
\end{table}

The \method{} conditional and joint NLLs are lower in all five pairs.
CoSMIC has lower learned model TV in four pairs. These comparisons describe the remaining
differences under the common training configuration. The original
method-specific comparisons are sensitive to choices of estimator and learning
rate. The original PHYLO result and the matched result use different cohorts
of optimizer starts, so they cannot identify a numerical share of the original
TV gap that is attributable to one optimizer component.

\subsubsection{Frozen-flow evidence diagnostics}
\label{app:primates-is-pilot}
A pilot using development-selected, method-specific flow rates
compares learned masses with direct importance sampling using the
frozen conditionals as proposals. In five pairs, four repeats per fit
use 512 draws per model each. Mean TV to pooled RJMCMC changes from
0.0873 to 0.0238 for \method{} and from 0.0514 to 0.0274 for CoSMIC.
The mean ranking reverses against one of the three reference chains.
Only two of the ten fits satisfy both the repeat and budget TV checks.

For the five common-training pairs in Table~\ref{tab:primates-additional},
we also compare direct and annealed importance sampling from the same
frozen proposals~\citep{neal2001ais}. AIS uses temperatures
$\beta_k=(k/K)^2$ and a preconditioned Crank--Nicolson update in the
Gaussian base chart~\citep{cotter2013pcn}. Four repeats at each of
$(K,\text{paths per model})=(16,64)$ and $(32,128)$ use 4,157,440 target
calls per fit. Direct IS matches that budget with four repeats of
5,120 draws per model. Model probabilities normalize evidence estimates
averaged on the original scale as in \eqref{eq:importance-model-marginal}.

Mean pooled-reference TV is 0.0157 versus 0.0217 for direct IS and
0.0252 versus 0.0258 for AIS, in \method{}/CoSMIC order. At equal calls,
direct IS passes the 0.05 repeat and budget-drift checks in all ten
fits. AIS fails the budget-drift check in all ten and its ranking varies
across reference chains. A model is flagged if its effective sample
size is below 20 or its largest normalized weight exceeds 0.25 in any
repeat. Every direct-IS fit still flags some models, representing
0.16--3.03\% of pooled reference mass across fits. Reference-block
resampling gives descriptive 2.5\% and 97.5\% quantiles of $-0.00782$
and $+0.00117$ for the mean direct-IS TV difference. These are exploratory
comparisons of the frozen proposals and their evidence estimates.
Table~\ref{tab:primates-additional} gives the results for the learned
model probabilities.

\subsection{Quantifying the effect of fixed conditionals on model mass}
\label{app:fixed-flow-audit}
For fixed $\phi$ and unrestricted model probabilities, put
$a_m=p(m)Z_m\exp\{-\delta_m(\phi)\}$ and $q^\star_m=a_m/\sum_j a_j$.
The objective difference is
$\mathcal L(q,\phi)-\mathcal L(q^\star,\phi)
=\KL(q\,\|\,q^\star)$. The remaining difference between $q^\star$ and the
target model marginal is governed by variation in $\delta_m(\phi)$.
We estimate the conditional free energies
$\E_{q_\phi}[\log q_\phi-\log\eta_m]$ with the flow frozen and apply the
analytic probability update. We evaluate the numerical stability of the
update with independent Monte Carlo replicates and nested sample sizes.

\paragraph{Fifteen-model benchmark with an exact target.}
We reuse all five fitted 1,500-update checkpoints and keep the continuous
flows fixed. Four new independent repeats draw 16,384 samples per model
each. The target supplies the normalized conditional density, so we can
estimate $\delta_m$ directly from the mean log-density ratio. The update
itself can equivalently use the prior and the conditional free energies,
with no target model masses as inputs. The two calculations agree. The
largest absolute difference from the first experiment in refitted TV to
the exact target is below $0.00044$ across the five fits.

All five refits reduce TV to the exact posterior. The largest
per-model Monte Carlo standard error for a conditional KL is below
0.003. Figures~\ref{fig:exact-gap-mass} and~\ref{fig:exact-gap-dimension}
show the model probabilities and conditional errors by dimension.
The dimension-3 mean gap exceeds the dimension-4 mean and gaps vary
within each dimension.

\begin{table}[H]
\centering
\small
\caption{\textbf{Learned and fixed-flow model probabilities.}
TV uses the exact posterior for the 15-model target and the same pooled
RJMCMC reference within each Medium data set. Policy KL is
$\KL(q_\psi\|q_\phi^\star)$, the decrease in joint reverse KL at fixed
continuous flows. TV, unlike the objective gap, has no additive decomposition.}
\label{tab:fixed-flow-comparison}
\begin{tabular}{@{}llrrrr@{}}
\toprule
Target & Updates & Learned TV & Fixed-flow TV & Negative ELBO & Policy KL \\
\midrule
Exact (5 starts) & 1,500 & 0.03155 & 0.00971 & -- & 0.00256 \\
\addlinespace
Medium (10 data sets) & 2,000 & 0.2049 & 0.2616 & 61.502 & 0.31994 \\
 & 30,000 & 0.2520 & 0.2521 & 60.973 & 0.00015 \\
\bottomrule
\end{tabular}
\par\smallskip
\footnotesize Entries are means, with median policy KL for Medium.
\end{table}

\paragraph{Changes during training on Medium.}
All twenty frozen-flow estimates pass the density and Monte Carlo
precision checks. At 2,000 updates the refit increases TV in nine of
ten fits. By 30,000 updates each learned distribution lies within
0.010 TV of the optimum for its own frozen conditionals.
Negative ELBO decreases in all ten fits while learned TV rises in
seven. Early discrete optimization error can partly offset the
fixed-flow tilt, and the continuous conditionals also change across
checkpoints. The tilt shifts target log odds for $m$ versus $j$ by
$\delta_j-\delta_m$. For the top two reference models per data set,
the median absolute tilt at 30,000 updates is 0.494, with a reversal
of their reference ordering in three fits.

\paragraph{Primates refits.}
Only three of twenty audited fit/checkpoint cases pass the precision
checks. In one precise 30,000-update \method{} refit, reference TV rises
from 0.0654 to 0.0804 while the objective improves by 0.00975.
This direction agrees across reference chains. The other seventeen
cases have insufficient precision for this comparison.

\begin{figure}[H]
\centering
\includegraphics[width=\textwidth]{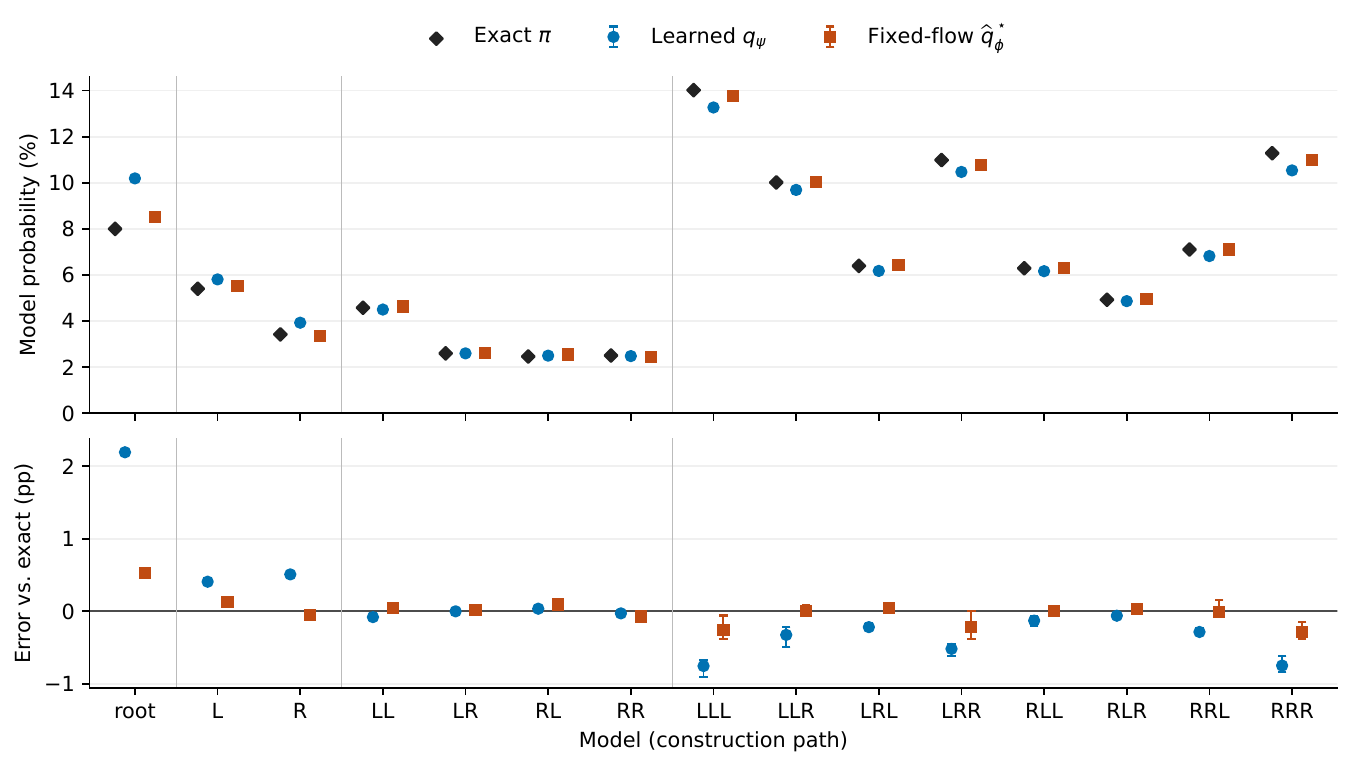}
\caption{\textbf{Model probabilities for the 15-model exact target.}
Black diamonds give the exact posterior. Coloured points give means of the
learned or fixed-flow probabilities across five fitted checkpoints, and
bars span their full range. The lower panel gives signed differences
from the exact probabilities in percentage points.}
\label{fig:exact-gap-mass}
\end{figure}

\begin{figure}[H]
\centering
\includegraphics[width=0.83\textwidth]{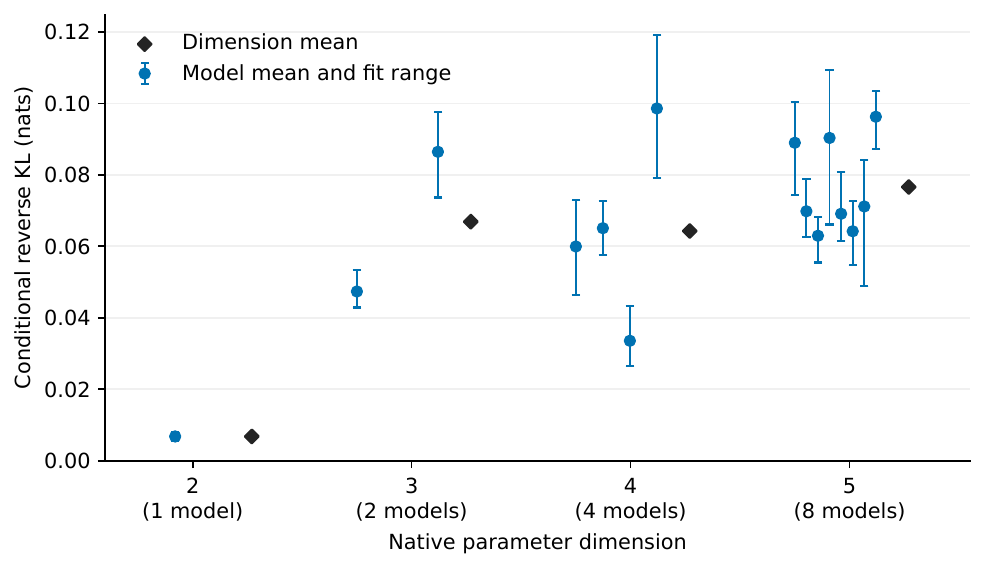}
\caption{\textbf{Conditional reverse KL by model dimension.}
Circles show per-model means across five fitted checkpoints. Vertical
bars span the five fitted values, with horizontal offsets to show
distinct models. Diamonds average the models in each dimension.
The numbers below the tick labels give the model counts.}
\label{fig:exact-gap-dimension}
\end{figure}

\subsection{Exact output-layer integration in small-sample DAGs}
\label{app:dag-output-integral}
For each non-root node, fix the hidden-layer coordinates and denote by
$\Phi_j$ the matrix of its ReLU features augmented with a column of ones.
With unit observation variance and a unit Gaussian prior on output-layer
coefficients, Bayesian linear marginalization~\citep{snoek2015output} gives
\[
  x_j\mid m,W_1
  \sim \mathcal N(0,I+\Phi_j\Phi_j^\top).
\]
One computes the log determinant and quadratic via the small coefficient
precision $I+\Phi_j^\top\Phi_j$. Root means, priors on the retained
coordinates, and the graph prior retain their original definitions.
Integrating the Gaussian output block preserves the marginal graph posterior.
The hidden-weight dimension still depends on the selected edges. When
needed, exact conditional Gaussian draws reconstruct the original output-block
coordinates. In particular, let $z$ be the retained variables and $\beta$
the integrated block. The chain rule gives
\[
\KL\{q(z)q(\beta\mid z)\,\|\,p(z,\beta\mid x)\}
=\KL\{q(z)\,\|\,p(z\mid x)\}
+\E_{q(z)}\KL\{q(\beta\mid z)\,\|\,p(\beta\mid z,x)\}.
\]
The last term vanishes for the exact conditional lift. Its precision
is $I+\Phi_j^\top\Phi_j$ and its mean solves the system with right-hand
side $\Phi_j^\top x_j$. This places the integrated ELBO on the same
target normalization as that of the full representation.

The four unguided arms cross the SM and VTI transports with the original
and output-marginalized continuous representations. They share the VTI
style score-function and information-limited policy update and the same
flow optimizer schedule. Each method retains its own graph-policy family.
The VTI marginal arm restores independent reference
factors in inactive coordinates for its saturated sampler.
We give the common 25,000-update endpoint on all six frozen data sets
(three per $N$, optimizer seed 3, 1,024 target calls per update).
Each fitted model yields 2,048 graph draws. This budget-specific experiment
tests both transports and the change of continuous representation.

\begin{table*}[t]
\centering
\small
\caption{\textbf{Small-sample DAG diagnostics after 25,000 shared-recipe updates.}
Each group is the mean across three fixed data sets.
Empty and dominant columns are empirical percentages over 2,048 draws
in each fit. The ELBO is scored for the same Bayesian model.}
\label{tab:dag-output-integral}
\begin{tabular}{@{}rllrrrrr@{}}
\toprule
$N$ & Family & Representation & Negative ELBO & Empty (\%) & Dominant (\%) & F1 & SHD \\
\midrule
16 & SM & Joint & 346.842 & 33.219 & 82.796 & 0.0980 & 24.545 \\
16 & SM & Integrated & 305.212 & 0 & 0.049 & 0.4539 & 22.889 \\
16 & VTI & Joint & 366.161 & 0 & 27.441 & 0.1150 & 24.804 \\
16 & VTI & Integrated & 311.186 & 0 & 0.049 & 0.4391 & 24.768 \\
\addlinespace
32 & SM & Joint & 622.268 & 33.285 & 80.355 & 0.1392 & 24.495 \\
32 & SM & Integrated & 563.450 & 0 & 0.602 & 0.5374 & 18.540 \\
32 & VTI & Joint & 651.304 & 0 & 50.537 & 0.1783 & 23.795 \\
32 & VTI & Integrated & 569.885 & 0 & 0.098 & 0.4975 & 21.181 \\
\bottomrule
\end{tabular}
\end{table*}

A lower bound on the total evidence, compared with the exact empty
graph evidence, gives $p(A=\varnothing\mid x)<4.5\times10^{-18}$ on all
six data sets. The calculation integrates the scalar output coefficients
and sums the original states with maximum indegree one to lower bound the
normalizer. Consequently the very high empirical empty frequencies in some
original arms are approximation errors. The integrated fits show no empty
draws at the displayed endpoint. Their full empirical model probabilities
remain estimates.

\paragraph{Learning trajectories and optimizer replications.}
The entropy-controlled update in the common recipe limits the change
in entropy of the encoding $M=(P,U)$. The adjacency $A$ is determined by
$M$, with $H_q(M)=H_q(A)+H_q(P\mid A)$. Several distinct orders can
therefore encode one dominant adjacency, as in the full VTI fits.

At 1k, 5k, 10k, and 25k, integrating the Gaussian outputs increases F1
and lowers negative ELBO in every within-method, within-data-set pair.
The same holds when repeating all four forms at $N=16$ with optimizer
seeds 17 and 29. Across the three optimizer seeds, the mean 25k F1s
are 0.069 to 0.456 for \method{} and 0.111 to 0.441 for VTI after
integration. Table~\ref{tab:dag-output-integral} uses seed 3 only,
to match the seven-sample-size comparison. SHD decreases for \method{}
in all six replication pairs at 25k and for VTI in three of six. The
integrated samples contain only nonempty graphs across these checkpoints.
A paired sparse-guide intervention with the same 1,024 calls per
update usually increases negative ELBO and lowers F1, so the primary
comparison uses the unguided fits.

\paragraph{Examining concentrated graph probabilities.}
Summing $q(P,U=0)$ over all $10!$ orders for the frozen 25k policy gives
empty-graph mass $0.99481$ for the full \method{} fit at $N=16$ on data
seed 1001. With the upper bound on target empty-graph mass above,
$\TV(q_A,p_A)\ge q_A(\varnothing)-p_A(\varnothing)>0.9948$.
The six integrated fitted masses at $N=16$ are below
$2.3\times10^{-11}$ and above their target upper bounds. Zero empty
draws at the evaluation sample size are therefore compatible with
remaining probability error. An analogous evidence bound for the
nonempty modes gives graph-TV lower bounds between 0.169 and 0.700
across the six full VTI fits, with simultaneous 95\% coverage for the
fitted graph probabilities.

We also compare fixed-flow odds for a sparse encoding and an empty
encoding with the same order and prior weight. The target
evidence bounds favor the sparse encoding. By \eqref{eq:conditional-gap-tilt},
the fixed-flow optimum favors the empty encoding in the six full
25k \method{} fits and the sparse encoding in the integrated fits.
This ordering also holds in the paired 5k diagnostic at $N=16$
with both transports and the two additional optimizer seeds.
These odds describe two order-and-edge encodings. The graph mass
sums over every compatible order.

\paragraph{Direct graph-posterior comparison on four nodes.}
On a separate synthetic problem with four nodes and one ReLU feature per
non-root node, we enumerate all 1,536 order-and-edge states, which give
543 labelled DAGs. We compute the fitted mass $q(A)$ by summing $q(m)$ over
the states $m$ that produce adjacency $A$. The reference uses the original
prior on order-and-edge states and sums their integrated likelihoods. For
each order it uses the zero-mean root likelihood and a product of non-root
local likelihood integrals. The Gaussian outputs integrate analytically.
Each remaining local integral has at most three hidden coordinates. We use
four scrambled Sobol estimates per integral, with $2^{18}$ points each at
the final level. The maximum pairwise TV between the resulting graph
references is below $1.1\times10^{-5}$. The pooled reference
changes by less than $9\times10^{-6}$ from the $2^{16}$ level.

Both families use 10,000 updates, 128 target calls per update, and
common optimizer settings at $N=16,32$, with three paired data sets per
size. The $N=16$ observations are a prefix of the corresponding
$N=32$ data. Graph TV and mean absolute edge-probability error decrease
after integration in all twelve within-method pairs at 1k, 5k and 10k.

For a fitted credible set of nominal mass $\alpha$, include graphs
in order of descending $q(A)$. Give all graphs at the boundary the
same inclusion probability $w_A$ so that $\sum_Aq(A)w_A=\alpha$.
Below the boundary $w_A=0$, and above it $w_A=1$. The reference mass
is $\sum_A\pi(A)w_A$. This handles ties and discrete overshoot.
Its absolute deviation from nominal mass decreases at each of the
levels 50\%, 80\%, 90\%, and 95\% in all twelve pairs.
At $N=32$ integrated CoSMIC has lower graph TV in all three pairs.
Integrated \method{} gives a reference mass closer to 90\% on average.

\begin{table}[H]
\centering
\small
\caption{\textbf{Four-node graph-posterior comparison at 10,000 updates.}
Entries are three-data-set means, shown from the full to the integrated
form. Edge error averages absolute probability error across 12 directed
edges. The last column gives reference posterior mass within fitted 90\%
credible sets, in percent.}
\label{tab:dag-graph-reference}
\begin{tabular}{@{}r l r r r@{}}
\toprule
$N$ & Method & Graph TV & Edge error & Reference mass at 90\% \\
\midrule
16 & \method{} & $0.174\to0.114$ & $0.0355\to0.0215$ & $81.67\to87.34$ \\
16 & CoSMIC & $0.355\to0.116$ & $0.0645\to0.0195$ & $62.26\to89.23$ \\
32 & \method{} & $0.300\to0.228$ & $0.0627\to0.0443$ & $73.93\to88.12$ \\
32 & CoSMIC & $0.651\to0.205$ & $0.1295\to0.0341$ & $35.27\to83.88$ \\
\bottomrule
\end{tabular}
\end{table}

\paragraph{Sensitivity to candidate-parent order.}
On the same six four-node problems we reverse the candidate-parent
visitation order within each child block of integrated \method{}.
The target, node-order choices, child-block sequence, capacity,
initial masses, optimizer, and budget are fixed. We train both policy
and transport under each order. The absolute change in TV error
relative to the same reference is below $0.00037$ in each pair.
Mean TV is $0.11443$ versus $0.11430$ at $N=16$ and $0.22775$ versus
$0.22797$ at $N=32$ for increasing and decreasing parent order.
The comparison concerns this specific change to the paths on four nodes.

\paragraph{Point-DAG decisions under common losses.}
\label{app:dag-point-decisions}
For a posterior method, let $\widehat q$ denote its empirical
distribution. Use the acyclic graph maximizing fitted expected F1
or minimizing fitted expected directed SHD. For a graph with $k>0$
edges, the expected-F1 edge weights are
\[
  w_e(k)=\sum_s \frac{2\,\Pr_{\widehat q}(e\in G,\ |G|=s)}{k+s}
\]
by the fixed-count decomposition~\citep{waegeman2014bayes}.
For SHD, the weights are $2\Pr_{\widehat q}(e\in G)-1$.
A subset dynamic program maximizes the score over acyclic graphs,
considering each possible last node and its incoming edges. Enumeration
on four nodes confirms both optima. For each integrated policy we pool
two independent banks of 8,192 graphs. The other policies or baselines
use their archived draws or weighted particles. DAGMA contributes its
point estimate. Reference graphs are used only for scoring in
Table~\ref{tab:dag-point-decisions}. DAGMA retains the highest reference
F1 at both small sample sizes under the expected-SHD rule.

\begin{table}[H]
\centering
\footnotesize
\setlength{\tabcolsep}{4pt}
\caption{\textbf{Reference scores for point DAGs selected under two fitted losses.}
Each entry is F1 / directed SHD. Synthetic rows average three paired
data sets. Each rule selects one DAG per distribution, and DAGMA
uses its original point output. The two VI methods use their
integrated 25,000-update fits at $N=16,32$.}
\label{tab:dag-point-decisions}
\begin{tabular}{@{}lrrrrrr@{}}
\toprule
$N$ or data & SM-VTI & VTI & DAGMA & DiBS & DiBS+ & JSP-GFN \\
\midrule
\multicolumn{7}{l}{\textit{Maximum fitted expected F1}} \\
16 & 0.593 / 18.33 & 0.610 / 21.67 & 0.474 / 34.00 & 0.211 / 25.00 & 0.222 / 23.33 & 0.405 / 27.00 \\
32 & 0.541 / 18.00 & 0.573 / 18.00 & 0.612 / 24.00 & 0.353 / 22.67 & 0.365 / 20.67 & 0.342 / 33.00 \\
64 & 0.436 / 18.33 & 0.338 / 22.00 & 0.659 / 19.00 & 0.473 / 20.67 & 0.448 / 19.67 & 0.426 / 34.33 \\
128 & 0.573 / 17.33 & 0.466 / 20.33 & 0.597 / 19.00 & 0.552 / 18.33 & 0.488 / 19.67 & 0.222 / 40.67 \\
256 & 0.665 / 15.33 & 0.515 / 21.00 & 0.603 / 16.67 & 0.680 / 15.00 & 0.554 / 18.33 & 0.316 / 45.00 \\
512 & 0.694 / 15.33 & 0.565 / 21.33 & 0.636 / 15.33 & 0.643 / 17.33 & 0.569 / 18.33 & 0.258 / 41.00 \\
1024 & 0.674 / 17.67 & 0.526 / 26.33 & 0.641 / 15.33 & 0.632 / 19.33 & 0.562 / 19.67 & 0.238 / 35.00 \\
Sachs & 0.471 / 18.00 & 0.444 / 25.00 & 0.286 / 25.00 & 0.293 / 29.00 & 0.333 / 28.00 & 0.281 / 41.00 \\
\midrule
\multicolumn{7}{l}{\textit{Minimum fitted expected directed SHD}} \\
16 & 0.432 / 21.00 & 0.379 / 20.67 & 0.474 / 34.00 & 0.025 / 25.33 & 0.222 / 23.33 & 0.408 / 26.67 \\
32 & 0.540 / 17.33 & 0.496 / 18.33 & 0.612 / 24.00 & 0.050 / 25.00 & 0.365 / 20.67 & 0.342 / 33.00 \\
64 & 0.436 / 18.33 & 0.338 / 22.00 & 0.659 / 19.00 & 0.097 / 24.33 & 0.448 / 19.67 & 0.422 / 34.00 \\
128 & 0.573 / 17.33 & 0.454 / 20.67 & 0.597 / 19.00 & 0.315 / 21.67 & 0.488 / 19.67 & 0.222 / 40.67 \\
256 & 0.665 / 15.33 & 0.523 / 20.33 & 0.603 / 16.67 & 0.337 / 20.67 & 0.554 / 18.33 & 0.316 / 45.00 \\
512 & 0.694 / 15.33 & 0.565 / 21.33 & 0.636 / 15.33 & 0.417 / 18.67 & 0.569 / 18.33 & 0.229 / 35.67 \\
1024 & 0.674 / 17.67 & 0.526 / 26.33 & 0.641 / 15.33 & 0.381 / 20.00 & 0.562 / 19.67 & 0.144 / 33.00 \\
Sachs & 0.471 / 18.00 & 0.444 / 25.00 & 0.286 / 25.00 & 0.250 / 18.00 & 0.333 / 28.00 & 0.281 / 41.00 \\
\bottomrule
\end{tabular}
\end{table}

\subsection{Sachs concentration diagnostics}
\label{app:sachs-integration}
The original fits in Table~\ref{tab:dag-sachs-all} concentrate on
different graphs across starts. Mean modal frequencies are 100\%
for \method{} and 98.14\% for VTI, from 2,048 draws per fit. Their
mean reference F1 is 0.496 and 0.361, respectively. These are distinct
observations about within-fit concentration and agreement with the
reference network.

\paragraph{Same-target probability bounds.}
We select a most frequent graph $A$ with the first 1,024 draws of
each original 25k fit. The other 1,024 give a one-sided exact binomial
lower bound on $q_A(A)$, with simultaneous 95\% coverage across the six
fits. For a distinct graph $B$, the target obeys
\[
   \pi_A(A\mid x)\leq\frac{Z(A)}{Z(A)+Z(B)},
\]
where $Z$ sums unnormalized masses over compatible orders. We test
adding one edge to a parentless node or deleting the sole parent of
a node. Unchanged local likelihood integrals cancel. The calculation
includes the original edge prior and the likelihood for the first
node in each order. Bounds on the remaining likelihoods, including
exact output integration and bounds over regions of the hidden
Gaussian prior, give informative odds in two starts. For \method{}
start 1 the target mass of $A$ is at most $5.5\times10^{-18}$,
and for VTI start 2 it is at most $2.9\times10^{-36}$. Combining
these bounds with the fitted probabilities gives graph TV above
0.995 and 0.978, respectively.

\paragraph{Analytic integration and the policy update.}
Within matched Adam pairs, integrating 60 Gaussian outputs lowers
the negative ELBO at both $5{,}000\times64$ and $1{,}000\times1{,}024$
update-by-draw budgets. F1 decreases from 0.463 to 0.186 in the former
(three starts) and from 0.531 to 0.167 in the latter (one start).
The integrated fits retain high modal frequencies, 99.59\% and 83.11\%.

We then pair integrated \method{} and CoSMIC under a common order-and-edge
policy, score-function SGD, AdamW for the flow, and 64 target calls per
update. CoSMIC's width is selected to match transport size (1,747,260
versus 1,748,180 continuous parameters and 129,164 policy parameters
for each). Within each method and start, the entropy-backtracking
intervention holds the target, initial weights, sampling seed, and
other training settings fixed. With backtracking, all 2,048 adjacency
draws are distinct in each fit at 1,000, 2,500, and 5,000 updates.
Without it, the leading adjacency accounts for at least 99.2\% of
draws in each fit at each checkpoint. Table~\ref{tab:sachs-entropy-switch}
shows the final checkpoint. The change in diversity measures the
effect of the policy update. Graph-posterior accuracy also depends
on the mass assigned to each adjacency.

\begin{table}[H]
\centering
\small
\caption{\textbf{Sachs policy-backtracking comparison at 5,000 updates.}
Means over three paired starts per transport. Modal frequency is the
share of 2,048 graph draws on the most frequent graph in each fit.
F1 and SHD use the reference network.}
\label{tab:sachs-entropy-switch}
\begin{tabular}{@{}llrrrr@{}}
\toprule
Transport & Backtracking & Modal frequency (\%) & Negative ELBO & F1 & SHD \\
\midrule
\method{} & On  & 0.049 & 109,298 & 0.1697 & 37.57 \\
\method{} & Off & 99.967 & 110,830 & 0.1370 & 44.67 \\
CoSMIC & On  & 0.049 & 110,217 & 0.1726 & 35.78 \\
CoSMIC & Off & 99.886 & 111,126 & 0.1150 & 23.67 \\
\bottomrule
\end{tabular}
\end{table}

\paragraph{Exact graph updates at fixed hidden marginals.}
For each of the twelve matched 1,000-update fits we use 64 draws of
the graph and hidden coefficients. Filling inactive coefficients
from their priors gives a labelled hidden bank $\bar H$. The integrated
target scores all $11\cdot2^{10}$ node/parent-set combinations for
each bank. A subset recursion sums over compatible orders
\citep{koivisto2004exact} and gives $\pi(A\mid\bar H,x)$.
Replacing the fitted graph conditional by this distribution at fixed
$q(\bar H)$ decreases joint reverse KL by
$\E_{q(\bar H)}\KL\{q(A\mid\bar H)\,\|\,\pi(A\mid\bar H,x)\}$.
We lower bound the decrease by $\E_q[-\log\pi(A\mid\bar H,x)]-110\log2$
using the entropy bound on 11-node graphs. Integrating a lower
survival envelope with a simultaneous 95\% Monte Carlo
bound~\citep{massart1990tight} gives a decrease of at least 4,948 for
each fit. Conditional resampling changes all 768 draw-level graphs.

In a separate check on three archived 50,000-update \method{} fits,
the same exact conditional changes all 384 graphs. The most frequent
graph accounts for 17.19--20.31\% of draws after the update, with
lower reference F1 and higher SHD in all three fits. The resampled
graph law averages the exact conditional over the fitted hidden
marginal. Accurate marginal graph probabilities also depend on
the quality of that hidden approximation.

\subsection{Quality and cumulative training time at 24 predictors}
\label{app:vs-quality-time}
Figure~\ref{fig:p24-time-quality} uses the actual active training times
recorded at three checkpoints on ten paired data sets in each setting.
Both original recipes use about 3.39 million learned parameters.
At 30,000 updates, \method{} takes longer per fit and has lower held-out
NLL in nine of ten pairs in both settings.

For each paired data set, we take VTI's best observed held-out joint NLL
over its 2,000-, 10,000-, and 30,000-update checkpoints as the quality
threshold. We compare the earliest saved SM checkpoint attaining that
threshold with the elapsed training time at the best VTI checkpoint.
On Medium, SM attains the threshold earlier in nine of ten pairs, all
at 2,000 updates. The remaining pair has no saved SM checkpoint attaining
it by 30,000 updates. On High, SM attains the threshold in seven of ten
pairs, and all seven attain it at a later recorded time. In each High pair
VTI's best saved checkpoint precedes 30,000 updates. These comparisons
use observed checkpoints and do not estimate crossing times between them.

\begin{figure}[!htbp]
\centering
\includegraphics[width=\textwidth]{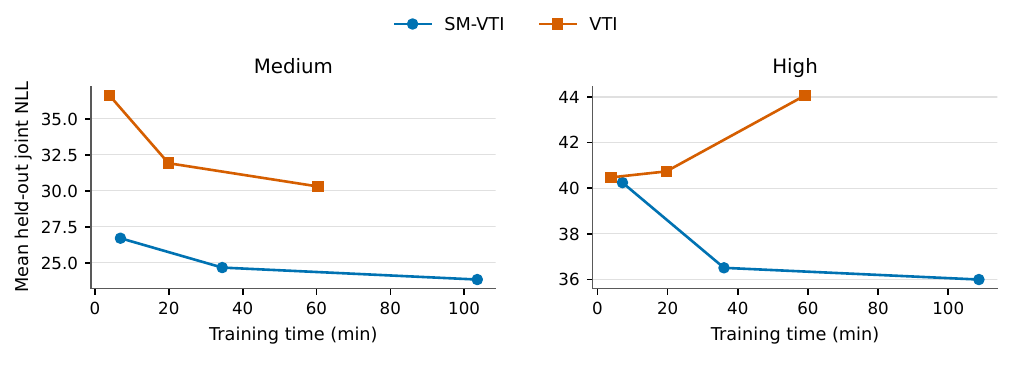}
\caption{\textbf{Observed quality and training time at 24 predictors.}
Each point averages ten paired data sets, and the three points correspond
to the 2,000-, 10,000-, and 30,000-update checkpoints. The horizontal axis
uses recorded active training time. The vertical axis uses held-out
RJMCMC-reference joint NLL.}
\label{fig:p24-time-quality}
\end{figure}

The original time records include a CUDA device index. Physical GPU model
within each pair is unavailable in the imported logs. These timings are
observed costs for the archived fits. The reference NLL uses 153,600
common retained RJMCMC states per data set.

\subsection{Comparison protocols across applications}
\label{app:comparison-protocols}
Table~\ref{tab:comparison-protocols} documents which choices are held fixed
in the original and supplemental fits. The SM-VTI and VTI pairs use the
same observations and target within each application. The exact output
integral changes the continuous representation while preserving the graph
marginal. The other DAG benchmarks use method-specific targets and training
budgets (Appendix~\ref{app:dag-baselines}).

\begin{table}[H]
\centering
\small
\setlength{\extrarowheight}{1pt}
\caption{\textbf{Training controls for the paired comparisons.}
Budgets are updates times target evaluations per update. Multiple starts
on a real alignment or on Sachs reuse the same data.}
\label{tab:comparison-protocols}
\begin{tabular}{@{}>{\raggedright\arraybackslash}p{.20\textwidth}
                   >{\raggedright\arraybackslash}p{.39\textwidth}
                   >{\raggedright\arraybackslash}p{.34\textwidth}@{}}
\toprule
\textbf{Setting} & \textbf{Shared training choices}
  & \textbf{Intervention or comparison} \\
\midrule
Medium and High VS, original finite fits
  & $30{,}000\times1{,}024$, similar parameter counts
  & Exact model sum for SM-VTI, score estimator for VTI \\
Medium VS, lift control
  & Same path, context, downstream width and initialization, optimizer, and $30{,}000\times1{,}024$
  & Learned lift versus append within SM-VTI on ten pairs \\
Medium and High VS, optimizer controls
  & Categorical policy, exact model sum, Adam settings, and $30{,}000\times1{,}024$
  & Three transports on ten Medium data sets and four on three High data sets \\
Medium and High VS, path-context control
  & Same policy, estimator, optimizer and budget in CPU float64, with common reference draws
  & Five transports on ten Medium and three High data sets, including complete-path CoSMIC at matched capacity \\
Primates, optimizer control
  & Categorical RLOO, Adam settings, flow-only warm-up, and $30{,}000\times256$
  & Five paired starts, with measured time and memory \\
Synthetic ten-node DAG, four-arm control
  & VTI-style updates and $25{,}000\times1{,}024$ at the common endpoint
  & Output integration within each family, three data sets each at $N=16,32$ \\
Sachs, original
  & $25{,}000\times1{,}024$ and endpoint selection criterion, three starts per method
  & Method-specific optimizers \\
Sachs, update control
  & Matching initializations within each representation and $5{,}000\times64$
  & Two update recipes crossed with full and integrated forms within SM-VTI \\
Vertebrates, original
  & $30{,}000\times256$ and uniform-policy warm-up, ten starts per method
  & Method-specific rates \\
\bottomrule
\end{tabular}
\end{table}

\paragraph{Common-optimizer controls.}
The original three-data-set High GPU comparison and the ten-data-set
Medium comparison use the same categorical initialization, exact
model sum, eight draws per conditional, and Adam at $10^{-3}$.
The High GPU control compares standard and wider CoSMIC, a wider
path-conditioned CoSMIC, and \method{}. Its mean model TVs are
0.1339, 0.1220, 0.1366, and 0.1192, respectively, with extreme
conditional NLL in some CoSMIC fits. Each data set's arms use the
same physical GPU. The standard CoSMIC and \method{} configurations
have 15,520 and 15,732 parameters, take 11.89 and 26.60 active
training minutes and allocate 26.73 and 42.36 MiB at peak.

On Medium, \method{} has lower conditional and joint NLL in six
of ten pairs against both standard and wider CoSMIC. The median
paired joint-NLL differences are $-0.161$ and $-0.009$, respectively.
The model-mass TV differences are smaller. The complete-path comparison
uses the same Medium controls and a separate set of High CPU fits.

\paragraph{CoSMIC conditioned on the complete construction path.}
For each of the 128 models, the path-conditioned CoSMIC takes a one-hot
model label and seven ordered path slots. An active slot supplies the
parent and child inclusion masks, the newly added coefficient, and a
validity bit. These features describe the same canonical edits used by
\method{}. In CoSMIC they condition the saturated flow. In \method{} the
edits determine the sequence of invertible maps, which models with a
common prefix share.

We set the widths of the two path-conditioned CoSMIC layers by minimizing
the absolute total parameter difference from \method{} over widths from
4 to 16 differing by at most one, with a fixed ascending layer order.
Widths 8 and 9 yield 15,666 total parameters, within 0.42\% of the 15,732
in \method{}. We locked this rule and choice before the new fits and
held-out evaluations. We also retain the original width-16 complete-path
arm and the standard and wider CoSMIC controls. All arms use the uniform
categorical initialization, 128-model exact sum, eight draws per model,
Adam at $10^{-3}$, and $30{,}000$ updates. The Medium control adds both
path variants to the three archived CPU arms. The High control runs all
five arms on CPU to match numerical precision and device type. Evaluation
uses the same 409,600 archived posterior draws for all arms on a data set.

Table~\ref{tab:vs-path-matched-capacity} retains all 65 endpoints, which
pass the inverse-density audit. Compared with the parameter-matched path
arm, \method{} has lower conditional and joint NLL in eight of ten
Medium pairs, with median differences $-0.438$ and $-0.428$,
respectively. Medium model TV is lower in five pairs, with a median
difference of $+0.00034$, and the matched path arm has lower model NLL
in seven pairs. On High, \method{} has lower conditional and joint NLL in
all three pairs against each CoSMIC variant. The median differences to
the matched path arm are $-2.037$ and $-2.100$. Both matched transports
flag the same High pair at 30k, and their differences have the same sign
on all three data sets. The observed NLL gains compare two finite-capacity
designs that share the same construction information, and the transports
differ in their native coordinates, lifts, and sharing of edge maps.

\begin{table}[H]
\centering
\small
\setlength{\tabcolsep}{4.5pt}
\caption{\textbf{CoSMIC with the complete path at matched and wider capacities.}
Scores are means over data sets at 30,000 updates. Flags count fits with
maximum held-out conditional NLL at least 1,000, at 30k or at any of the
five saved steps. All flagged fits are included.}
\label{tab:vs-path-matched-capacity}
\begin{tabular}{@{}lrrrrrrr@{}}
\toprule
Transport & Params. & Model TV & Model NLL & Cond. NLL & Joint NLL & 30k flag & Any flag \\
\midrule
\multicolumn{8}{@{}l}{\textit{Medium (10 pairs)}} \\
SM-VTI & 15,732 & 0.2517 & 4.403 & 1.262 & 5.665 & 0/10 & 0/10 \\
CoSMIC & 15,520 & 0.2546 & 4.346 & 1.830 & 6.176 & 1/10 & 1/10 \\
Path CoSMIC, matched & 15,666 & 0.2598 & 4.368 & 2.158 & 6.526 & 0/10 & 0/10 \\
Wider CoSMIC & 29,728 & 0.2499 & 4.336 & 1.404 & 5.741 & 0/10 & 1/10 \\
Wider Path CoSMIC & 30,304 & 0.2539 & 4.349 & 1.816 & 6.165 & 0/10 & 1/10 \\
\midrule
\multicolumn{8}{@{}l}{\textit{High (3 pairs)}} \\
SM-VTI & 15,732 & 0.1205 & 4.118 & 2.303 & 6.421 & 1/3 & 2/3 \\
CoSMIC & 15,520 & 0.1422 & 4.124 & 3.592 & 7.717 & 0/3 & 2/3 \\
Path CoSMIC, matched & 15,666 & 0.1923 & 4.187 & 4.552 & 8.739 & 1/3 & 2/3 \\
Wider CoSMIC & 29,728 & 0.1229 & 4.109 & 2.730 & 6.839 & 0/3 & 2/3 \\
Wider Path CoSMIC & 30,304 & 0.1427 & 4.130 & 3.250 & 7.380 & 0/3 & 1/3 \\
\bottomrule
\end{tabular}
\end{table}

\FloatBarrier
\section{Measured runtime and memory}
\label{app:runtime}

\paragraph{Original configurations on a fixed GPU.}
Table~\ref{tab:runtime-1000} summarizes three synchronized timing
repetitions on one NVIDIA L40S in float64. Each restores a saved
endpoint and measures 1,000 optimizer updates after warm-up, with
1,024 target calls per update for variable selection and DAGs and 256
for phylogenetics. Timing excludes setup, warm-up, and evaluation.
The DAG measurements restore one 25,000-update checkpoint per
method and include policy and flow updates. Their original methods
use different parameter counts and training recipes. The added DAG
baselines use the budgets in Table~\ref{tab:dag-baseline-protocol}.

\begin{table}[H]
\centering
\small
\caption{\textbf{Update cost for the original experimental configurations.}
Times are mean seconds per 1,000 updates. The measured peak GPU allocation
is in GiB. Separate end-to-end time and quality comparisons are given
in Appendix~\ref{app:vs-quality-time}.}
\label{tab:runtime-1000}
\begin{tabular}{@{}lrrrrrr@{}}
\toprule
& \multicolumn{2}{c}{Parameters (millions)}
& \multicolumn{2}{c}{Time (s)}
& \multicolumn{2}{c}{GPU allocation (GiB)} \\
\cmidrule(lr){2-3}\cmidrule(lr){4-5}\cmidrule(l){6-7}
Setting & SM-VTI & VTI & SM-VTI & VTI & SM-VTI & VTI \\
\midrule
Medium VS, $p=7$  & 0.0158 & 0.0155 & 23.55 & 17.67 & -- & -- \\
Medium VS, $p=9$  & 0.1267 & 0.1278 & 26.05 & 57.55 & -- & -- \\
Medium VS, $p=24$ & 3.3908 & 3.3890 & 213.18 & 122.75 & -- & -- \\
High VS, $p=24$   & 3.3908 & 3.3890 & 218.55 & 124.79 & -- & -- \\
\addlinespace
Primates     & 0.0930 & 0.1240 & 124.30 & 81.75 & -- & -- \\
Vertebrates  & 0.1072 & 0.3201 & 139.94 & 96.67 & -- & -- \\
\addlinespace
DAG $N=64$   & 2.8849 & 54.2094 & 472.78 & 494.41 & 3.37 & 2.06 \\
DAG $N=256$  & 2.8849 & 54.2094 & 326.44 & 498.18 & 3.61 & 2.27 \\
Sachs        & 2.1536 & 206.9815 & 1167.66 & 1625.58 & 9.80 & 13.34 \\
\bottomrule
\end{tabular}
\end{table}

In the separate shared-optimizer, output-integrated comparison
against standard-size CoSMIC (2,427,555 versus 54,209,365 parameters),
mean active times through 25k are 186.00 versus 353.73 minutes at $N=16$
and 176.60 versus 342.88 at $N=32$. Paired runs share a physical GPU
within each stage. One continuation bundle per size flags concurrent
use by 25k.

\paragraph{Small-sample DAG comparison at similar parameter counts.}
We also choose a CoSMIC reduction in hidden widths by parameter count
while retaining its flow blocks and masks. Each of three $N=16$
pairs shares the integrated target, optimizer, L40S, and 25,000
updates of 1,024 calls. We measure 100 synchronized updates after
12 warm-up steps for early time and peak allocation. Active time
excludes setup, evaluation, and checkpointing.

\begin{table}[H]
\centering
\small
\caption{\textbf{Near-parameter-matched DAG cost on three $N=16$ data sets.}
Entries are means. Time to comparison quality uses the first saved
checkpoint attaining CoSMIC's best recorded negative ELBO within each
pair. GPU allocation is the peak over the early 100-update window.}
\label{tab:dag-matched-size-trajectory}
\begin{tabular}{@{}lrr@{}}
\toprule
Measurement & SM-VTI & Size-matched CoSMIC \\
\midrule
Learned parameters & 2,427,555 & 2,424,445 \\
Early seconds per 100 updates & 55.25 & 57.68 \\
Peak GPU allocation (GiB) & 4.201 & 0.499 \\
Active hours to comparison quality & 0.683 & 4.033 \\
Active hours through 25k & 3.146 & 4.033 \\
Negative ELBO at 25k & 305.212 & 311.154 \\
Mean F1 at 25k & 0.454 & 0.439 \\
\bottomrule
\end{tabular}
\end{table}

CoSMIC attains its best saved negative ELBO at 25k, which \method{}
first reaches at the 5k checkpoint in each pair. At 25k \method{}
has higher F1 in one of three pairs and higher measured GPU
allocation in every pair.

\paragraph{Real-data training costs.}
Table~\ref{tab:primates-additional} reports matched Primates costs.
For the six integrated 1k Sachs checkpoints, a 32-update window after
eight warm-up steps on the same two AMD EPYC 9754 cores gives mean
CPU times of 344.37 versus 342.07 seconds per 100 updates
(\method{}/CoSMIC). Peak RSS including setup and warm-up is
1.76 versus 1.66 GiB. The protocol uses 64 target calls per update.

\end{document}